\documentclass[twoside]{article}
\usepackage{qic,epsfig}

\usepackage{latexsym}
\usepackage[]{graphics,graphicx}            
\usepackage{bbold,enumerate,mathtools}
    \usepackage{hyperref}
\newtheorem{definition}{Definition}
\usepackage[numbers]{natbib}
\newcommand{\ket}[1]{\left | #1 \right\rangle}
\newcommand{\bra}[1]{\left \langle #1 \right |}
\newcommand{\half}{\mbox{$\textstyle \frac{1}{2}$}}
\newcommand{\smallfrac}[2][1]{\mbox{$\textstyle \frac{#1}{#2}$}}
\newcommand{\Tr}{\text{Tr}}
\newcommand{\braket}[2]{\left\langle #1|#2\right\rangle}
\newcommand{\proj}[1]{\ket{#1}\bra{#1}}
\newcommand{\identity}{\mathbb{1}}
\renewcommand{\epsilon}{\varepsilon}
\begin{document}

\setlength{\textheight}{8.0truein}    

\runninghead{Optimal Universal Quantum Cloning: Asymmetries and Fidelity Measures}
            {Alastair Kay}

\normalsize\textlineskip
\thispagestyle{empty}
\setcounter{page}{1}

\copyrightheading{0}{0}{2003}{000--000}

\vspace*{0.88truein}

\alphfootnote

\fpage{1}

\centerline{\bf
Optimal Universal Quantum Cloning: Asymmetries and Fidelity Measures}
\vspace*{0.035truein}
\centerline{\bf FOR QUANTUM INFORMATION AND COMPUTATION}
\vspace*{0.37truein}
\centerline{\footnotesize
ALASTAIR KAY}
\vspace*{0.015truein}
\centerline{\footnotesize\it Department of Mathematics, Royal Holloway University of London,}
\baselineskip=10pt
\centerline{\footnotesize\it Egham, Surrey, TW20 0EX, United Kingdom}
\vspace*{10pt}
\vspace*{0.225truein}
\publisher{(received date)}{(revised date)}

\vspace*{0.21truein}

\abstracts{
We study the problem of universal quantum cloning -- taking several identical copies of a pure but unknown quantum state and producing further copies. While it is well known that it is impossible to perfectly reproduce the state, how well the copies can be cloned can be quantified using the fidelity. We examine how individual fidelities can be traded against each other, and how different fidelity measures can be incorporated. The broadly applicable formalism into which we transform the cloning problem is described as a series of quadratic constraints which are amenable to mathematical and computational scrutiny. As such, we reproduce all known results on optimal universal cloning, and push the recent results on asymmetric cloning much further, giving new trade-off relations between fidelities for broad classes of optimal cloning machines. We also provide substantial evidence that motivates why other parameter ranges (number of input copies) have not, and will not, yield to similar analysis.
}{}{}

\vspace*{10pt}

\vspace*{3pt}
\communicate{to be filled by the Editorial}

\vspace*{1pt}\textlineskip    

\section{Introduction}

Quantum cloning is the quintessential no-go theorem of quantum mechanics -- possible in the classical world, but impossible to implement perfectly in the quantum world, and reflecting such fundamental properties of the quantum world that can be used as a postulate in information theoretic explorations. Ever since the original proofs that an unknown quantum state cannot be perfectly cloned \cite{wootters1982}, quantifying how well states can be cloned has proved a challenge. When all the clones are required to be the same quality, the achievable qualities are well understood \cite{buek1996,werner1998,gisin1997,bruss1998}. This covers the case not only of universal cloning, in which the input state is equally likely to be any pure state, but also state-dependent cloning, of qudits \cite{scarani2005}. However, asymmetric cloning, when the clones are permitted to have different qualities, has proven far more challenging. Until recently, studies were limited to very specific cases of, for example $1\rightarrow 3$ universal cloning of qubits \cite{iblisdir2005,iblisdir2006}, in which one input copy is converted to 3 output copies of differing qualities. However, a recent revelation has permitted calculation of the trade-offs in $1\rightarrow N$ universal cloning of qudits, for arbitrary $N$ and local Hilbert space dimension $d$ \cite{kay2009-a,kay2013}. In addition, they revealed a more fundamental insight; there is a direct connection between the ability to share correlations between different spins and the quality of clones that can be produced on those spins. These monogamy-type relationships provide widely applicable bounds for the study of strongly correlated quantum systems, elevating interest in asymmetric cloning beyond that of mathematical curiosity to the foundations of a powerful new calculational tool with properties different to those encapsulated by monogamy of the tangle \cite{coffman2000,osborne2006} or of Bell tests \cite{ramanathan2011}, with the added benefit that, by knowing the {\em optimal} cloning results, these bounds are incredibly stringent.

In the past, \cite{kay2009-a} has asserted a result for the cloning trade-off relations of $1\rightarrow N$ cloning of qudits (spins with a $d$-dimensional Hilbert space), as measured by the 1-copy fidelity, but relied on an unproven assumption. The result was proven more rigorously  in \cite{kay2013}, and extended to the case of $N-1\rightarrow N$ cloning for qubits. This paper details the full proof of both $1\rightarrow N$ and $N-1\rightarrow N$ cloning of qudits, using not only the 1-copy fidelity, but the $L$-copy fidelity \cite{wang2011}. Indeed, for $N-1\rightarrow N$ cloning, any arbitrary combinations of different fidelities can be used. We also study the case of $M\rightarrow N$ cloning when $M\neq 1,N-1$, providing strong reasons why these cases do not give similar answers and, instead, are anticipated to require an inefficient non-convex optimisation to solve in generality.

\subsection{Motivation}

One scenario in which asymmetric quantum cloning is important, as are the different fidelities that we explore in this paper, is eavesdropping of quantum key distribution. Of course, a true security proof has to be more general than a single eavesdropping strategy, but this is a useful illustration. Consider Alice and Bob trying to share a secret key by a protocol such as BB84 \cite{bennett1984} or E91 \cite{ekert1991}. We'll consider a device independent variant of E91 for the sake of concreteness \cite{acin2007}. We can imagine an eavesdropper, Eve, intercepting each qubit that travels to Bob, cloning it, and allowing one copy to make its way to Bob. Meanwhile, Eve holds the second copy and waits for the announcement of measurement bases before making her measurement. There are three different fidelities in this scenario; $F_B$ (the quality of Bob's clone), $F_E$ (the quality of Eve's clone) and $F_{B,E}$ (the joint quality of the clones). It turns out that $F_B$ is the probability that Bob gets the ``correct'' measurement result (the one he should have got if he had received the qubit he was supposed to), while $F_{B,E}$ is the probability that both Bob and Eve get the correct result.

Now, imagine that Bob is performing some tests, in coordination with Alice, on his measurement results, and detects his success rate, $F$. For a given answer, how much might Eve know, i.e.\ how much privacy amplification will Alice and Bob need to perform later? If Eve were to perform asymmetric cloning, she could set $F_B=F$ (under the device independence assumption, Eve has the power to replace Bob's measurement apparatus with perfect measuring devices). Her success rate is $F_E$, which she would optimise; this is the study of optimal asymmetric cloning. However, more relevant, is to optimise the probability of her getting the ``correct'' measurement result given that Bob did, which is $F_{B,E}/F_B$. Thus we need to maximise $F_{B,E}$ for fixed $F_B$. This is optimal asymmetric cloning with mixed fidelities, which is covered by the formalism of this paper. In fact, it turns out that in this case, $F_{B,E}$ and $F_E$ are directly related ($F_{B,E}=F_E+F_B-1$, Eq.\ (\ref{eqn:N-1})), but that is not obvious {\em a priori}. In Figure \ref{fig:QKD} we plot the outcome of the relevant optimal cloning study, assuming cloning has been performed universally, meaning Alice and Bob would need to be using arbitrary measurement bases. Alternatively, a finite set of bases is possible provided they comprise a 3-design \cite{ambainis2007} (meaning that taking averages over three copies of a finite set of states yields identical results to the averaging over all possible states). The standard basis choices of E91 and BB84 are only 1-designs. For these cases, there are higher fidelity cloning options because we only have to concentrate on a subset of states to clone, specifically equatorial cloning \cite{kay2013}.

\begin{figure}
\centerline{\includegraphics[width=0.45\textwidth]{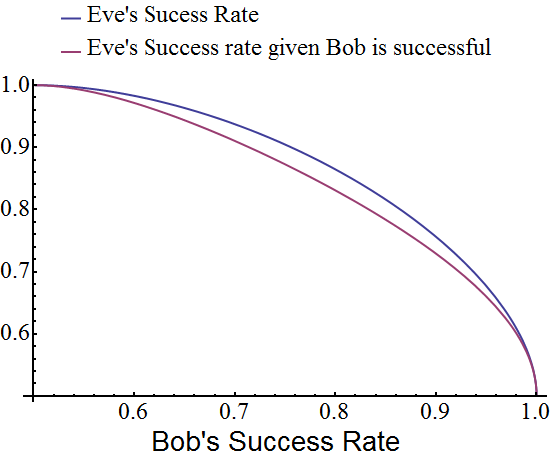}}
\fcaption{\label{fig:QKD} When eavesdropping a QKD scheme by cloning, one can trade the risk of detection with the amount of information gained about the key.}
\end{figure}

\subsection{Orientation}

This is a mathematically detailed paper that makes heavy use of specific notation which is introduced in Section \ref{sec:notation}. This could usefully be used as a reference while reading the rest of the document. Having introduced this notation, we state the main results of the paper in \ref{sec:summary} and explore some consequences of these in Sec.\ \ref{sec:applications} -- Sec.\ \ref{sec:N-1} completely solves the problem of when $N-1$ input copies can be transformed into $N$ copies of varying qualities, while Sec.\ \ref{sec:1} studies the problem of transforming a single input state into $N$ outputs.

The proof of the main results appears in Sec.\ \ref{sec:proof}. It starts with Sec.\  \ref{sec:prelim}, some preliminary properties of the symmetric subspace, and how these properties are altered by the application of the partial transpose operation. These are likely to be familiar to many readers. In Sec.\ \ref{sec:isomorphism}, we review the Choi-Jamio\l kowski isomorphism \cite{jamiokowski1972}, which is the main technical tool that allows us to give an upper bound on the achievable cloning fidelity by calculating the maximum eigenvalue of a certain matrix. It is then in Section \ref{sec:clone} that we apply this isomorphism to the cloning problem. In particular, we isolate a small subspace of the relevant matrix and use a variant of the Lieb-Mattis theorem (Sec.\ \ref{sec:liebmattis}) to show that the maximum eigenvalue of the whole matrix is contained within this subspace \cite{lieb1962} and furthermore confirm that the upper bound specified by this formalism can indeed be attained. The resultant problem to be solved consists of a set of quadratic constraints that need to be satisfied. Our formalism also allows us to make simple statements about economical cloning -- the question of whether $M\rightarrow N$ cloning only requires the $N$ output spins, or whether ancilla spins (`anti-clones') are required, which depends on the values of $M$ and $N$.

The study of different numbers of input copies appears to be far more challenging than the cases of $M=1$ and $N-1$. Sec.\ \ref{sec:otherM} proves, aside from these special cases, the problem is inherently non-convex, implying that these cases are far more challenging than the previously solved special cases. Indeed, a solution is not expected (unless P=NP).

For completeness, Sec.\ \ref{sec:circuits}, briefly describes how the optimal asymmetric universal quantum cloner can be realized before we conclude in Sec.\ \ref{sec:conclusion}. Appendix A provides some potentially valuable side results -- properties of the key matrices that we do not utilize directly, but could be useful in the future, for improving the implementation protocols for instance.

\subsection{Notation} \label{sec:notation}

We will often need to specify operators and states of qudits as they apply to particular sets or subsets of spins. When it is unimportant {\em which} spins are being acted upon (perhaps we discuss the generic properties of an $N$-spin operator), we use a sub/superscript $N$. Otherwise, we use the sub/superscript to specify the qubits being acted upon. For cloning, as studied in this paper, there will be two distinct sets of spins; $M$ input spins which we collectively denote by IN, and $N$ output spins, which are collectively denoted by OUT. In order to refer to subsets of the output spins (and this is always the output spins, not the input), we use bit strings $x\in\{0,1\}^N$ (we reserve letters $x$, $y$ and $z$ for such bit strings). These have a Hamming weight (number of 1 entries) $w_x=x\cdot x$. The bit string $x$ divides the output spins into the two sets specified by sites $\{n: x_n=1\}$ and $\{n: \bar x_n=1\}$ where $\bar x$ is the complement of $x$; $P^x$ means that the operator $P$ acts on the output sites $n$ for which $x_n=1$. We can also combine sets of sites: $x\cup y$ conveys the set of sites for which either $x_n=1$ or $y_n=1$, while $x\cap y$ restricts to those for which $x_ny_n=1$.

\begin{definition}
When cloning $M$ identical copies of an input state $\ket{\phi}$, the output of the cloner is an $N$ qubit state $\rho_{\phi}$. If we interpret $\proj{\phi}^0=\identity$, then the cloner has fidelities
$$
F_y=\Tr\left(\rho_{\phi}\bigotimes_{n=1}^N\proj{\phi}^{y_n}\right)
$$
for $y\in\{0,1\}^N$. This measures the fidelity of a subset of sites defined by $y_n=1$. $F_y=1$ means that all sites where $y_n=1$ are in the state $\ket{\phi}$, and have been perfectly cloned.
\end{definition}

This is the cloning fidelity for a specific input state $\ket{\phi}$. Universal cloning, as studied in this paper, involves averaging this fidelity over all possible input states. The solutions that we derive will have the same fidelity for all possible input states.

In the case where the fidelities we are interested in correspond to all the bit strings of weight $L$, we say that we are examining the $(M,L,N)$ cloning machines.

\begin{definition}
Define the matrices
$$
G_y^{(M)}=\sum_{\substack{z,x\in\{0,1\}^N\\ w_x=w_z=M}}\frac{\ket{x}\bra{z}}{\binom{M+d-1-x\cdot z+w_{\bar x\cap\bar z\cap y}}{d-1}}.
$$
A specific cloning problem has fixed $M$, so the superscript can safely be dropped. 
 $d$ is a positive integer. These may be generalized to
$$
G_y^{(M,L)}=\sum_{\substack{x\in\{0,1\}^N\\ w_x=L\\ x\cdot y=\min(L,w_y)}}G^{(M)}_x.
$$
\end{definition}
Included in this definition are the special cases $G_{\underline 0}^{(M)}$ and $G_{\underline 0}^{(M,L)}$, which have a common structure that enables simple solution for eigenvalues and eigenvectors (see Appendix A).

\subsection{Summary of Results}\label{sec:summary}

The main technical results that we prove in this paper are as follows.

\begin{theorem}\label{thm:main}
An asymmetric quantum cloning machine of $M\rightarrow N$ qudits, of local Hilbert space dimension $d$, producing output fidelities $F_y$ for $y\in\Lambda$ is optimal if and only if there exists $\underline{\beta}\in\mathbb{R}^{\binom{N}{M}}$ such that $\underline{\beta}^T\cdot G_{\underline{0}}^{(M)}\cdot\underline{\beta}=1$ and $F_y=\underline{\beta}^T\cdot G_{y}^{(M)}\cdot\underline{\beta}$ for all $y\in\Lambda$.
\end{theorem}

\begin{corollary}\label{cor:main}
Asymmetric cloning is possible provided
$$
\min_{\stackrel{\underline{\beta}\in\mathbb{R}^{\binom{N}{M}}}{F_y\leq\underline{\beta}^T\cdot G_{y}^{(M)}\cdot\underline{\beta}\forall y\in\Lambda}}\underline{\beta}^T\cdot G_{\underline{0}}^{(M)}\cdot\underline{\beta}\leq 1.
$$
\end{corollary}
This is essentially a matter of definition; we assert that if the task is to clone with a certain set of fidelities, we'd be happy if we actually achieved a higher fidelity. Aside from that, Cor.\ \ref{cor:main} is just a restatement of Thm.\ \ref{thm:main}, but in a manner more suited to standard computational techniques for the resolution of the question.

\begin{theorem}\label{thm:N-1}
If $M=N-1$, then the quadratic constraints of Cor.\ \ref{cor:main} reduce to linear ones for any specified set $F_y$ where $y\in\Lambda$.
\end{theorem}
This is proven in Sec.\ \ref{sec:N-1}.

\begin{theorem}\label{thm:constraints}
If $\Lambda=\{y\in\{0,1\}^N:w_y=L\}$ then
\begin{itemize}
\item For $(1,1,N)$ and $(1,N-1,N)$ cloning, all constraints in Cor.\ \ref{cor:main} become linear.
\item Otherwise, for $(1,L,N)$ cloning, all but $\half N(N-3)$ quadratic constraints become linear.
\item If $M=2,3,\ldots N-2$, and $2M<L<N-2M$, then the $\binom{N}{L}$ quadratic constraints of Cor.\ \ref{cor:main} can be reduced to $\binom{N}{2M}$ quadratic constraints, and the rest linear.
\end{itemize}
\end{theorem}
If all the conditions are linear, the system can be efficiently solved. While some quadratic constraints remain, efficient solution is unlikey, but the reduced problem wherein the quadratic constraints are removed can be solved efficiently and can witness the impossibility of cloning. The first item is proven in Sec.\ \ref{sec:1}, while the rest are the subject of section \ref{sec:linear}.

\section{Applications of Theorem \ref{thm:main}}\label{sec:applications}

\subsection{Symmetric Cloning}

Consider the case of symmetric cloning, in which we demand that all fidelities be measured on the same number of spins (say $L$) and that all fidelities be the same. If we set $\underline{\beta}$ to be uniform then all the $\underline{\beta}^T\cdot G_{y}^{(M)}\cdot\underline{\beta}$ are equal (the $G_y$ are identical under permutations). Hence
\begin{eqnarray*}
F_y&=&\underline{\beta}^T\cdot G_{y}^{(M)}\cdot\underline{\beta}	\\
&=&\frac{1}{\binom{N}{L}}\sum_{y\in\Lambda}\underline{\beta}^T\cdot G_{y}^{(M)}\cdot\underline{\beta}	\\
&=&\frac{1}{\binom{N}{L}}\underline{\beta}^T\cdot G_{\underline 0}^{(M,L)}\cdot\underline{\beta}.
\end{eqnarray*}
In Appendix A, we show that the maximum eigenvalue of $G_{\underline 0}^{(M,L)}$ is given by
$$
\frac{1}{\binom{N}{L}}\sum_{i=0}^M\binom{M}{i}\binom{N-M}{M-i}\sum_{q=0}^{N+i-2M}\frac{\binom{N+i-2M}{q}\binom{2M-i}{L-q}}{\binom{M+d-1-i+q}{d-1}},
$$
and that our choice of $\underline{\beta}$ coincides with the maximum eigenvector. Taking into account the normalization of the state, $\underline{\beta}^T\cdot G_{\underline{0}}^{(M)}\cdot\underline{\beta}$ yields
\begin{equation}
F=\frac{1}{\binom{N}{L}}\frac{\binom{M+d-1}{M}}{\binom{N+d-1}{M}}\sum_{i=0}^M{\textstyle\binom{M}{i}\binom{N-M}{M-i}}\!\!\sum_{q=0}^{N+i-2M}\frac{\binom{N+i-2M}{q}\binom{2M-i}{L-q}}{\binom{M+d-1-i+q}{d-1}},	\label{eqn:sym1}
\end{equation}
which simplifies to
\begin{equation}
F=\frac{1}{\binom{N}{L}\binom{N+d-1}{N-M}}\sum_{i=0}^M\sum_{q=0}^N\textstyle{\binom{M}{i}\binom{q-M}{i}\binom{N-M+d-1}{N-q}\binom{M+i}{q-L}}.	\label{eqn:sym2}
\end{equation}
Note that we define $\binom{a}{b}=0$ if $b<0$ or $b>a$. This means that, depending on the values of $L$ and $M$, many terms may be eliminated from the sum. In the special case of $L=1$, the sum over $q$ in Eq.\ (\ref{eqn:sym1}) is restricted between 0 and 1, giving
$$
F=\frac{M}{N}+\frac{(N-M)(M+1)}{N(M+d)},
$$
which coincides with the standard result \cite{werner1998,scarani2005}. In the case of $L=N$, the sum in Eq.\ (\ref{eqn:sym2}) is restricted to $q=N$, giving the known result for the global fidelity,
$$
F=\frac{\binom{M+d-1}{M}}{\binom{N+d-1}{N}}.
$$
While Wang et al. \cite{wang2011} state an equivalence to Eq.\ (\ref{eqn:sym2}) of
$$
F=\sum_{q=0}^N\frac{\binom{q}{M}\binom{q}{L}\binom{N-q+d-2}{d-2}}{\binom{N}{L}\binom{N+d-1}{N-M}},
$$
we have been unable to prove that equivalence beyond special cases and numerical tests.

\subsubsection{Simplified Fidelity Tests}

The symmetric cloner already provides a simple test for whether certain cloning tasks are achievable or not. Given that, for $(M,L,N)$-cloning, we are asking whether there exists a suitably normalised vector such that
$$
F_y\leq \beta^TG_y^{(M)}\beta\qquad\forall y:w_y=L,
$$
then by summing these, we have
$$
\sum_{y: w_y=L}F_y\leq \beta^TG_{\underline 0}^{(M,L)}\beta\leq \binom{N}{M}F_{\text{sym}}.
$$
So, if this inequality is violated, cloning must be impossible regardless of the asymmetry. Obviously, this is tight at the point of perfect symmetry, and will be a good approximation close to that point. Also, simple considerations in the case study of Appendix \ref{sec:casestudy} suggest that it could often be the case for non-trivial $M$ that the bound is tight over much broader ranges.
Evidently, satisfying the constraint can never be sufficient proof that cloning is possible as it disguises all the subtleties of the desired asymmetries. As a clear example, consider $3\rightarrow 4$ cloning with the $L=2$ fidelity. It is possible to achieve 3 fidelities all being 1 (e.g.\ $F_{1100},F_{0110},F_{1010}$) because this just corresponds to perfect teleportation from the input spins to the first 3 output spins. Whereas, demanding $F_{1100}=F_{0011}=F_{1010}=1$ will clearly be impossible and yet it could still have the same total fidelity, see Fig.\ \ref{fig:3to4}.

\begin{figure}
\centerline{\includegraphics[width=0.45\textwidth]{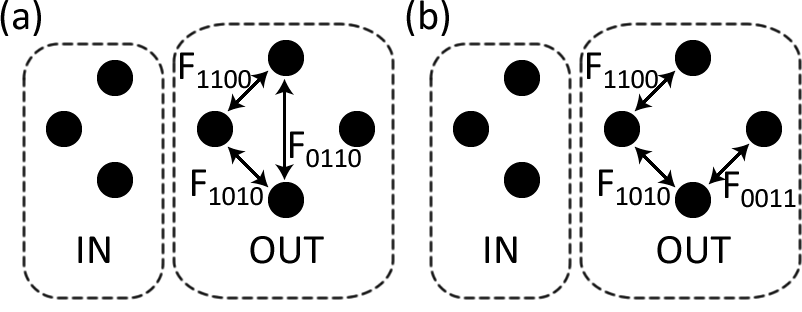}}
\fcaption{\label{fig:3to4}When tasked with converting 3 identical input states into 4 copies, it is possible to get 3 copies perfectly, and hence certain subsets of 2-copy fidelities can all be 1 (a). Meanwhile, other subsets (b) would require every copy to be identical to the input state, which is clearly impossible.}
\end{figure}

\subsection{Fidelity Relations: \texorpdfstring{$N-1\rightarrow N$}{N-1 to N} cloning} \label{sec:N-1}

When $M=N-1$, there is a single site at which a given bit string $x\in\{0,1\}^N: w_x=N-1$ (referring to $\beta_x$) is 0. So, we choose (in this section only) to revise the notation to $\beta_n$ where $x_n=0$, while if $L=1$ we use $F_k$ to mean $y_k=1$ (the only site at which $y$ is not 0). The normalization condition rearranges as
$$
\left(\sum_n\beta_n\right)^2+(d-1)\sum_n\beta_n^2= d.
$$
The fidelities are given by
$$
F_k=1-\frac{d-1}{d}\beta_k^2,
$$
permitting elimination of the $\beta_k$ from the normalisation condition, and yielding the optimal fidelity trade-off for $(N-1,1,N)$-cloning as
\begin{equation}
N-1=\sum_nF_n-\frac{1}{d-1}\left(\sqrt{1-F_n}\right)^2.
\end{equation}
Moreover, for arbitrary $L$, we can readily find that
\begin{equation}
F_y=1-\frac{d-1}{d}\sum_{n:y_n=1}\beta_n^2=\sum_{n:y_n=1}F_n-L+1=y\cdot\underline{F}-w_y+1	\label{eqn:N-1}
\end{equation}
where $\underline{F}$ is the vector of single-copy fidelities. Given a set of fidelities, we solve for $\underline{F}$:
$$
\min_{y\cdot\underline{F}\geq F_y+w_y-1 \forall y\in\Lambda}\sum_nF_n-\frac{1}{d-1}\left(\sum_n\sqrt{1-F_n}\right)^2.
$$
Cloning is possible if and only if this value is $\leq N-1$. We note that this is a convex optimization problem -- most obviously because the target function encapsulates the information from the positive definite matrix $G_{\underline 0}^{(M)}$ (that this is positive definite is clear because $G_{\underline 0}^{(M)}$ represents the Gram matrix of the vectors $\ket{\psi_x}$. Nevertheless, a full calculation of the eigenvalues is given in Appendix A). Since the problem is one of convex optimization (see also Lemma \ref{lem:convex}), it can be efficiently solved by interior point methods \cite{boyd2004}. Thus, the cloning problem can be resolved if $M=N-1$ for any set $\Lambda$. This proves Theorem \ref{thm:N-1}.

\subsection{\texorpdfstring{$1\rightarrow N$}{1 to N} cloning} \label{sec:1}

When $M=L=1$, the bit strings have a single site at which there is a 1, so we replace $x\in\{0,1\}^N: w_x=1$ with a value $n\in[N]$. The normalization condition reads
$$
\left(\sum_n\beta_n\right)^2+(d-1)\sum_n\beta_n^2= d.
$$
Each fidelity is calculated as
$$
d(d+1)F_k=d+\left((d-1)\beta_k+\sum_n\beta_n\right)^2.
$$
By summing over all $k$, we get that
$$
\frac{d(d+1)}{N+d-1}\sum_kF_k=2\left(\sum_n\beta_n\right)^2+(d-1)\sum_n\beta_n^2,
$$
and we can rearrange the fidelity relation to give
$$
(N+d-1)\sum_n\beta_n=\sqrt{d}\sum_n\sqrt{F_n(d+1)-1}.
$$
Eliminating these from the normalization condition yields
\begin{equation}
\frac{(d+1)\sum_nF_n}{N+d-1}=1+\left(\frac{\sum_n\sqrt{F_n(d+1)-1}}{N+d-1}\right)^2,
\end{equation}
describing the optimum trade-off between the cloning fidelities. This relation is equivalent to that found for the special case of $N=3$ \cite{iblisdir2005}, and was subsequently verified (at a time when Corollary \ref{cor:important} was only a conjecture) for $N=4$ \cite{ren2011, wikliski2012}. Similar expressions can be derived for any value of $L$, as detailed in Sec.\ \ref{sec:1p2}. If we set all the fidelities equal, we have
$$
(N+d-1)(d+1)NF=(N+d-1)^2+N^2((d+1)F-1),
$$
which rearranges to
$$
F=\frac{2N+d-1}{(d+1)N},
$$
recovering the standard result on symmetric cloning.

\subsection{A Case Study: \texorpdfstring{$2\rightarrow 4$}{2 to 4} Cloning}\label{sec:casestudy}
 We illustrate the challenges involved in solving a case where $M\neq 1,N-1$ by examining in more detail the first such case, $2\rightarrow 4$ cloning, using the two-copy fidelity measure $L=2$ on qubits, $d=2$ (The choice of $L=M$ can be anticipated to be the easiest to solve, as one would hope to find a one-to-one function between the $\{\beta_x\}$ and the $\{F_y\}$). 

Let $x,y,z\in\{0,1\}^4$ be bit strings of weights $w_x=w_y=w_z=2$, and let $\underline{\beta}=\sum_x\beta_x\ket{x}$ be a vector of real numbers. Lemma \ref{lem:nolinear} will prove that there are no linear combinations of the matrices which are rank 1. However, if we assume that
$$
F_y=F_{\bar y}	\qquad \forall y,
$$
the task is significantly simplified.
This is equivalent to
\begin{equation}
\bra{\beta}G_y-G_{\bar y}\ket{\beta}=0	\qquad \forall y.	\label{eqn:conditions}
\end{equation}
Applying the basis change specified by
$$
\tilde H=\frac{1}{\sqrt{2}}\sum_{\substack{x\in\{0,1\}^3\\w_x=1}}\!\!(\proj{x1}-\proj{\bar x 0}+\ket{x1}\bra{\bar x 0}+\ket{\bar x 0}\bra{x1})
$$
gives, for example,
\begin{eqnarray*}
\tilde H G_{\underline 0}\tilde H &=&\frac{1}{30}\left(
\begin{array}{cccccc}
 16 & 15 & 15 & 0 & 0 & 0 \\
 15 & 16 & 15 & 0 & 0 & 0 \\
 15 & 15 & 16 & 0 & 0 & 0 \\
 0 & 0 & 0 & 4 & 0 & 0 \\
 0 & 0 & 0 & 0 & 4 & 0 \\
 0 & 0 & 0 & 0 & 0 & 4
\end{array}
\right)	\\
\tilde H(G_{0011}-G_{1100})\tilde H &=&\left(
\begin{array}{cccccc}
 0 & 0 & 0 & 0 & 0 & \frac{2}{15} \\
 0 & 0 & 0 & 0 & 0 & \frac{1}{10} \\
 0 & 0 & 0 & 0 & 0 & \frac{1}{10} \\
 0 & 0 & 0 & 0 & -\frac{1}{10} & 0 \\
 0 & 0 & 0 & -\frac{1}{10} & 0 & 0 \\
 \frac{2}{15} & \frac{1}{10} & \frac{1}{10} & 0 & 0 & 0
\end{array}
\right).	
\end{eqnarray*}
There are 4 different ways to satisfy the three conditions of Eq.\ (\ref{eqn:conditions}),
$$
\ket{\beta}=\left\{\begin{array}{c}	(a,b,c,0,0,0)^T \\
(-\frac{3}{4}(a+b),a,b,0,0,c)^T	\\
(a,-\frac{3}{4}(a+b),b,0,c,0)^T	\\
(a,b,-\frac{3}{4}(a+b),c,0,0)^T
\end{array}\right..
$$
We have to try each of these cases in turn in order to assess which can achieve the largest fidelities, although the symmetry between the last 3 cases means we only have to assess one of them. First, consider the case $\ket{\beta}=(a,b,c,0,0,0)^T$, which reduces to $3\times 3$ matrices.
\begin{eqnarray*}
\bra{\tilde\beta}\left(
\begin{array}{ccc}
 16 & 15 & 15 \\
 15 & 16 & 15 \\
 15 & 15 & 16
\end{array}
\right)\ket{\tilde\beta}&=&30	\\
\frac{22}{30}\bra{\tilde\beta}\left(
\begin{array}{cccccc}
 1 & 0 & 0 \\
 0 & 0 & 0\\
 0 & 0 & 0
\end{array}
\right)\ket{\tilde\beta}&=&13(F_{1100}+F_{0011})+9
-9(F_{0101}+F_{1010}+F_{1001}+F_{0110})
\end{eqnarray*}
where $\ket{\tilde\beta}=(a,b,c)^T$. $\ket{\tilde\beta}$ is immediate, and the normalization condition is satisfied if
\begin{equation}
\sqrt{2(F_{1100}+F_{1010}+F_{0110})-1}=
\sqrt{3}\sum_{\substack{x\in\{0,1\}^3\\w_x=2}}\sqrt{44F_{x0}-18(F_{1100}+F_{1010}+F_{0110})+9},	\label{eqn:class2}
\end{equation}
although there are also constraints on the accessible values, since, for example, $13(F_{1100}+F_{0011})+9-9(F_{0101}+F_{1010}+F_{1001}+F_{0110})\geq 0$ because it corresponds to a squared quantity.

\begin{figure}
\centerline{\includegraphics[width=0.45\textwidth]{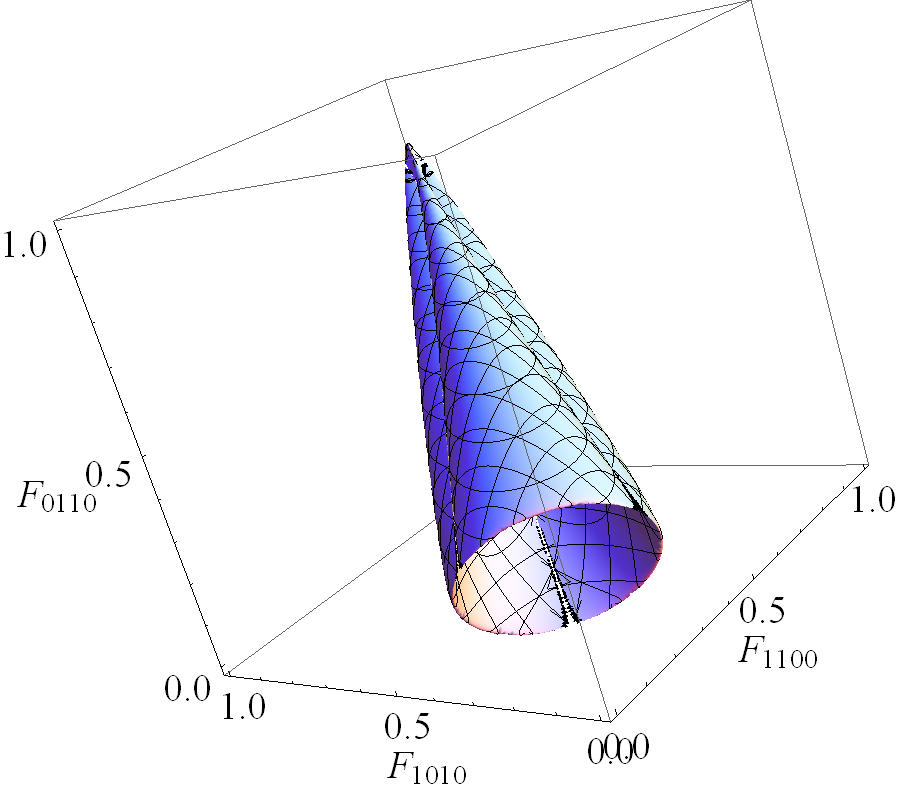}}
\fcaption{\label{fig:class2}Surface described by optimal cloning relation Eq.\ (\ref{eqn:class2}).}
\end{figure}

The above relation, depicted in Fig.\ \ref{fig:class2}, describes what, at first glance, might appear a surprising region -- as one fidelity increases, the other fidelities also increase! However, this actually makes sense; if $F_{0011}$ and $F_{1100}$ are to both be high, they need a large component of an entangled state going from the two input spins to both qubit pairs $(1,2)$ and $(3,4)$, which automatically means that all fidelities must be high. In fact, it turns out that within this class of solutions, the fully symmetric point is a global maximum, i.e.\ there's no point in worrying about asymmetric cloning because the best result is to always implement fully symmetric cloning, and each of the fidelities is $\frac{61}{69}$.

This must be compared to $\ket{\beta}=(-\frac{3}{4}(a+b),a,b,0,0,c)^T$, which similarly yields
\begin{equation}
(2F_{1010}+2F_{0110}+2-7F_{1100})(2F_{1010}+2F_{0110}-1-F_{1100})=\frac{9}{2}(F_{1010}-F_{0110})^2,	\label{eqn:class1}
\end{equation}
with positivity constraints:
\begin{eqnarray*}
2F_{1010}+2F_{0110}-1-F_{1100}&\geq& 0	\\
2F_{1010}+2F_{0110}+2-7F_{1100}&\geq &0	\\
2F_{1010}+2F_{0110}-1-\frac{5}{3}F_{1100}&\leq& 0.
\end{eqnarray*}


\begin{figure}
\centerline{\includegraphics[width=0.45\textwidth]{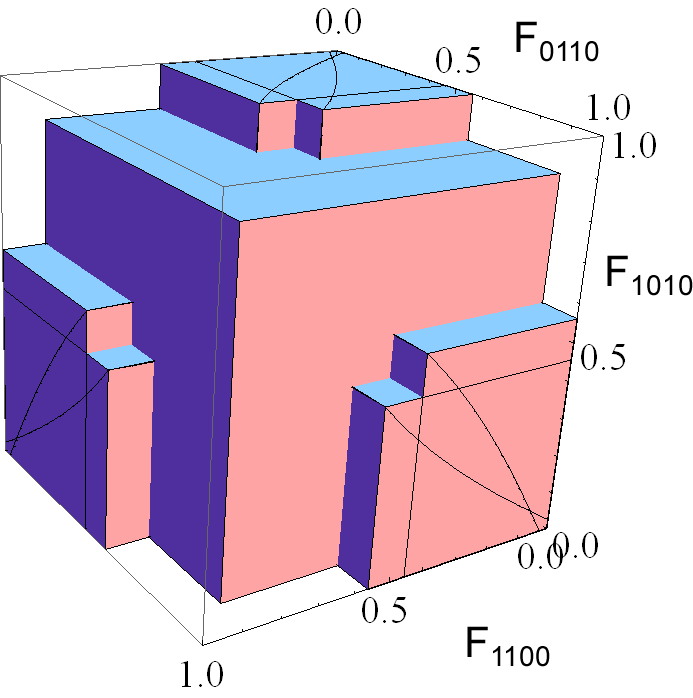}}
\fcaption{\label{fig:final} Fidelity triples inside the shaded region can be achieved or exceeded.}
\end{figure}

We say that a set of cloning fidelities $\{F_{1100},F_{1010},F_{0110}\}$ can be achieved if there exists another set of fidelities $\{\tilde F_{1100},\tilde F_{1010},\tilde F_{0110}\}$ for which $\tilde F_x\geq F_x$ for all $x$, such that the $\tilde F_x$ sit on the optimal cloning surfaces. The achievable fidelities are thus plotted in Fig.\ \ref{fig:final}, where we see there are distinct phases of different cloning results, which is in stark contrast to the $1\rightarrow N$ cloning that has been studied in the past in which there was no phase transition in the system as the cloning fidelities are varied.

\section{Proof of Theorem \ref{thm:main}}\label{sec:proof}

\subsection{Preliminaries -- The Symmetric Subspace and Spin Operators} \label{sec:prelim}

Our proofs are closely connected to the symmetric subspace and its properties.
\begin{definition}
Let $\mathcal{S}_N$ be the symmetric group on $N$ letters, and $\pi\in\mathcal{S}_N$ be an arbitrary permutation of those letters. The symmetric subspace of a Hilbert space $\mathcal{H}=(\mathbb{C}_d)^{\otimes N}$ is the set of states on $\mathcal{H}$ that are invariant under all such permutations:
$$
S_{\mathcal{H}}:=\{\ket{\psi}\in\mathcal{H}:P(\pi)\ket{\psi}=\ket{\psi}\forall\pi\in\mathcal{S}_N\}
$$
where $P(\pi)$ is a representation of $\pi$ on $\mathcal{H}$.
\end{definition}
\begin{definition}Projector onto the symmetric subspace:
$$
P_{\text{sym}}^{N,d}=\frac{1}{N!}\sum_{\pi\in\mathcal{S}_N}P(\pi).
$$
In this context, the superscript conveys that the projector acts on $N$ spins of dimension $d$. Later, it will also be used to specify which subset of spins it acts on.
\end{definition}
\noindent To see that this is a projector, note that $P(\pi_1)P(\pi_2)=P(\pi_1\pi_2)$, $\pi_1\pi_2$ is in the group, and the group multiplication is invertible, meaning that $\pi_1\pi_2$ maps to distinct permutations for all $\pi_2$ and a fixed $\pi_1$. Hence
$$
{P_{\text{sym}}^{N,d}}^2=\frac{1}{N!N!}\sum_{\pi_1\in\mathcal{S}_N}\sum_{\pi_2\in\mathcal{S}_N}P(\pi_1)P(\pi_2)=\frac{1}{N!}\sum_{\pi\in\mathcal{S}_N}P(\pi)=P_{\text{sym}}^{N,d}.
$$
Evidently, for any state $\ket{\psi}\in S_{\mathcal{H}}$, $P_{\text{sym}}^{N,d}\ket{\psi}=\ket{\psi}$, so the span of states that is projected onto includes the symmetric subspace. Furthermore, for any $\ket{\psi}\in\mathcal{H}$ and $\pi\in\mathcal{S}_N$,
$
P(\pi)P_{\text{sym}}^{N,d}\ket{\psi}=P_{\text{sym}}^{N,d}\ket{\psi}
$
because, again, we can fold the $P(\pi)$ into the sum over $\pi$ in the projector, and this is the definition of a state in $S_{\mathcal{H}}$. Hence the span of states that is projected onto is a subspace of the symmetric subspace. Taken together, this shows that $P_{\text{sym}}^{N,d}$ projects onto $S_{\mathcal{H}}$.

A useful feature of the symmetric subspace is its relation with the spin operators:
\begin{definition}
For a $d$-dimensional Hilbert space, we define the following spin operators:
\begin{eqnarray*}
S^{X,d}\!&=&\!\sum_{n=1}^{d-1}\!\sqrt{n(d-n)}(\ket{n-1}\bra{n}+\ket{n}\bra{n-1})	\\
S^{Y,d}\!&=&\!\sum_{n=1}^{d-1}\!i\sqrt{n(d-n)}(\ket{n}\bra{n-1}-\ket{n-1}\bra{n})	\\
S^{Z,d}\!&=&\!\sum_{n=0}^{d-1}(d-1-2n)\proj{n}.
\end{eqnarray*}
For an $N$-fold tensor product of such a Hilbert space, the total spin operators are
$$
J_K=\sum_{n=1}^N\identity^{\otimes (n-1)}\otimes S^{K,d}\otimes\identity^{\otimes(N-n)},
$$
with $K=X,Y,Z$, and $J^2=J_X^2+J_Y^2+J_Z^2$.
\end{definition}
It is important to note that $[J_Z,J^2]=0$, so they are simultaneously diagonalizable. We will use $\{\ket{\phi_i^N}\}$ to denote an orthonormal basis for $P_{\text{sym}}^{N,d}$, i.e.
$$
P_{\text{sym}}^{N,d}=\sum_i\proj{\phi_i^N},
$$
and will generally assume that the $\ket{\phi_i}$ are also eigenstates of the $J_Z$ operator (for qubits, for example, one can fix that $J_Z\ket{\phi_i^N}=(2i-N)\ket{\phi_i^N}$ for $i=0,\ldots, N$). Alternatively, we will use the superscript to denote the set of spins that the state covers. For instance, we may use $\ket{\phi^x_i}$ to denote a symmetric state of $w_x$ spins located on the sites $n$ specified by $x_n=1$. It is only a mild abuse of notation to then write $\ket{\phi_i^x}\ket{\phi_j^{\bar x}}$ to fully specify the state of $N$ spins. We can also create a Bell state from the symmetric states, defining
$$
\ket{B_x^{(M)}}=\frac{1}{\sqrt{\binom{M+d-1}{M}}}\sum_i\ket{\phi_i^{\text{IN}}}\ket{\phi_i^x}
$$
for any $x\in\{0,1\}^N: w_x=M$. Lemma \ref{lem:dimension} will confirm that this state is correctly normalized.

\subsubsection{Properties of the Symmetric Subspace}

We now review some basic properties of the symmetric subspace \cite{harrow2013} (itself a synthesis of works including \cite{barenco1997,christandl2006,goodman2009,harrow2005,jozsa1998,stanley2001,chiribella2010}), and describe the effect of the partial transpose operation.

\begin{lemma} \label{lem:dimension}
The dimension of $S_{\mathcal{H}}$ is $\binom{N+d-1}{N}$.
\end{lemma}
\begin{proof}{
The dimension of the symmetric subspace can be calculated from $\Tr(P_{\text{sym}}^{N,d})$. If we use $[d]$ to denote a choice of labels $1,2,\ldots, d$ then
$$
\Tr(P_{\text{sym}}^{N,d})=\frac{1}{N!}\sum_{\pi\in\mathcal{S}_N}\sum_{i\in[d]^N}\bra{i}P(\pi)\ket{i}.
$$
We can consider this sum as, for each $i$, how many permutations are there that map $i$ to $i$? If there are $c_1$ instances of 1 in $i$, $c_2$ of 2 etc. (subject to the constraint $\sum_jc_j=N$), then there are $c_1!$ permutations that map all the 1s back to the 1s. Hence,
$$
\Tr(P_{\text{sym}}^{N,d})=\frac{1}{N!}\sum_{i\in[d]^N}\prod_{j=1}^Nc_j!.
$$
But, of all the strings $i$, how many have $c_1$ 1s, $c_2$ 2s etc? $\frac{N!}{\prod_{j=1}^Nc_j!}$. This leaves us needing to know the number of distinct configurations of the $\{c_j\}$ that are possible, i.e.\ how many ways are there to distribute $N$ indistinguishable items between $d$ bins? $\binom{N+d-1}{N}$.
}\end{proof}

\begin{lemma} \label{lem:PT}
The operator defined on $M$ input qudits and $w$ output qudits by
$$
\rho_{(\text{IN},w)}:=\int (U^*\proj{0}U^T)^{\otimes M}\otimes (U\proj{0}U^\dagger)^{\otimes w} dU,
$$
with integration being taken uniformly over the Haar measure for $U\in SU(d)$, satisfies
$$
\rho_{(\text{IN},w)}^{T_{\text{IN}}}=\frac{P_{\text{sym}}^{M+w,d}}{\binom{M+w+d-1}{M+w}}.
$$
$T_{\text{IN}}$ denotes the partial transpose over the $M$ input spins.
\end{lemma}
\begin{proof}{
If we take the partial transpose, we have that
$$
\rho_{(\text{IN},w)}^{T_{\text{IN}}}=\int (U\proj{0}U^\dagger)^{\otimes (M+w)} dU.
$$
This is clearly a mixture of all possible states $\proj{\psi}^{\otimes (M+w)}$, which is the symmetric subspace, we just have to be careful with the normalization. The trace is unaffected by partial transpose operations, so given that $\Tr(\rho_{(\text{IN},w)})=1$ and, by Lemma \ref{lem:dimension},
$$
\Tr(P_\text{sym}^{M+w,d})=\binom{M+w+d-1}{M+w},
$$
we have the desired result.
}\end{proof}

\begin{corollary}\label{lem:positive}
The matrix elements of $\rho_{(\text{IN},w)}$ are non-negative.
\end{corollary}
\begin{proof}{
An element $N!\bra{i}P_{\text{sym}}^{N,d}\ket{j}$ counts the permutations that map the string $i$ in to the string $j$. This is clearly non-negative. The partial transpose rearranges matrix elements and does not change their values.
}\end{proof}

\noindent Lemma \ref{lem:PT} contains the statement of twirling \cite{horodecki1999} as a special case (with, perhaps, a more straightforward proof):
\begin{corollary}\label{lem:twirl}
For $M=w=1$, $\rho_{(1,1)}=\frac{\proj{B^{(1)}}}{d+1}+\frac{\identity}{d(d+1)}$.
\end{corollary}
\begin{proof}{
The basis elements of the symmetric subspace of a $d\times d$ Hilbert space consist of $\ket{ii}$ and $(\ket{ij}+\ket{ji})/\sqrt{2}$ for $i< j$.
$$
P_{\text{sym}}^{2,d}=\sum_{i=0}^{d-1}\proj{ii} +\half\sum_{i<j}(\ket{ij}+\ket{ji})(\bra{ij}+\bra{ji}),
$$
so
\begin{equation*}
\begin{split}
\rho_{(1,1)}&=\frac{2}{d(d+1)}{P_{\text{sym}}^{2,d}}^{T_\text{IN}}	\\
&=\frac{2\!\displaystyle\sum_i\!\proj{ii}+\!\displaystyle\sum_{i<j}\!\ket{ij}\!\bra{ij}+\ket{ji}\!\bra{ji}+\ket{ii}\!\bra{jj}+\ket{jj}\!\bra{ii}}{d(d+1)}	\\
&=\frac{\identity}{d(d+1)}+\frac{\proj{B^{(1)}}}{d+1}.
\end{split}
\end{equation*}
}\end{proof}

\begin{lemma}\label{lem:commutator}
The matrix $\rho_{(IN,w)}$ satisfies
\begin{eqnarray*}%
\left[\rho_{(IN,w)},(U_I^{\otimes M}\otimes\identity^{\otimes w})J_Z({U_I^\dagger}^{\otimes M}\otimes\identity^{\otimes w})\right]&=&0	\\
\left[\rho_{(IN,w)},(U_I^{\otimes M}\otimes\identity^{\otimes w})J^2({U_I^\dagger}^{\otimes M}\otimes\identity^{\otimes w})\right]&=&0
\end{eqnarray*}
where $U_I=\sum_{n=0}^{d-1}(-1)^n\ket{n}\bra{d-1-n}$.
\end{lemma}
\begin{proof}{
It is clear that $[P_{\text{sym}}^{M+w,d},J_Z]=0$ (and for $J_X, J_Y$) and $[P_{\text{sym}}^{M+w,d},J^2]=0$ given that the total spin operators are invariant under permutations of underlying spins. Now, divide the sum for $J_Z$ (for instance) into a sum over terms on the input space, and terms on the output space. It must be that
$$
\left[P_{\text{sym}}^{M+w,d},\sum_{\text{IN}}S^Z+\sum_{\text{OUT}}S^Z\right]=0,
$$
so we take the partial transpose over the input spins, using the fact that $[A,B]^T=-[A,B]$:
$$
\left[\rho_{(IN,w)},-\sum_{\text{IN}}{S^Z}^T+\sum_{\text{OUT}}S^Z\right]=0.
$$
Furthermore,
$-S_Z^T=-S_Z$, $-S_X^T=-S_X$ and $-S_Y^T=S_Y$, so given that the specified $U_I$ maps $S_Z\mapsto -S_Z$, $S_X\mapsto -S_X$ and $S_Y\mapsto S_Y$, it must be that
$$
[\rho_{(IN,w)},(U_I^{\otimes M}\otimes\identity^{\otimes w}J_Z{U_I^\dagger}^{\otimes M}\otimes\identity^{\otimes w}]=0.
$$
}\end{proof}

\begin{lemma}\label{lem:transpose}
For any $d\times d$ matrix $M$
$$
M^T\otimes\identity\ket{B^{(1)}}=\identity\otimes M\ket{B^{(1)}}.
$$
\end{lemma}
The proof of this is a simple case of explicitly writing $M=\sum_{i,j}M_{i,j}\ket{i}\bra{j}$ and verifying the equivalence. It is left to the reader.

\subsection{The Choi-Jamio\l kowski Isomorphism} \label{sec:isomorphism}

A general scenario that encompasses quantum cloning is that, given one of a set of states $\ket{\psi_i}$ ($i = 1, \dots, K$), we are required to perform a particular state transformation on it, without being told which of the states we have been given. The required transformation may not be achievable exactly within the quantum formalism, but is best approximated within the theory by a completely positive, trace preserving map $\mathcal{E}$ that transforms input state $\ket{\psi_i}$ into $\mathcal{E}(\proj{\psi_i})$. The success of the state transformation task is then measured by a fidelity 
\begin{equation*}
F=\frac{1}{K}\sum_i\Tr\left(\mathcal{M}_i\mathcal{E}(\proj{\psi_i})\right). 
\end{equation*}
Here $\mathcal{M}_i$ are positive operators ($\mathcal{M}_i \succeq 0$) satisfying $\|\mathcal{M}_i\|\leq 1$ so that $F$ is indeed a fidelity taking values between $0$ and $1$. If the fidelity takes value $1$, we infer that the map has perfectly implemented the required state transformation for all the specified input states. For example, if required to transform the states $\ket{\psi_i}$ into states $\ket{\phi_i}$, we might define $\mathcal{M}_i=\proj{\phi_i}$. For a continuous set of states, the sum appearing in the definition of $F$ transforms to an integral. The factor of $1/K$ appearing in the definition stems from the assumption that each of the input states $\ket{\psi_i}$ is equally likely. If this is not the case, these parameters can be adjusted based on a given probability distribution of the input states.
\begin{lemma}\label{lem:transformation}
For the state transformation task, the achievable fidelity is upper bounded by the maximum eigenvalue of the operator
\begin{equation}
R=\frac{d'}{K}\sum_i\proj{\psi_i}^T_I\otimes \mathcal{M}_i, \label{eqn:def_R}
\end{equation}
where $d'$ is the dimension of the subspace spanned by the states $\{\ket{\psi_i}\}$.
\end{lemma}
\begin{proof}{
We start by introducing a Hilbert space of two parts, an input space and an output space, both of dimension $d$, the dimension of the Hilbert space from which the $\ket{\psi_i}$ are taken, and consider the maximally entangled state $\ket{\Psi}_{IN,OUT}$ between them, as applies only to the subspace  of states $\{\ket{\psi_i}\}$. This lets us rewrite $$\proj{\psi}_{\text{OUT}}=d'\Tr_{\text{IN}}(\proj{\psi}^T\otimes\identity\cdot(\identity\otimes\mathcal{E})(\proj{\Psi}))$$ using Lemma \ref{lem:transpose}. So,
$$
F=\frac{d'}{K}\sum_i\Tr\left(\proj{\psi}^T\otimes \mathcal{M}_i\cdot(\identity\otimes\mathcal{E})(\proj{\Psi})\right).
$$
Since $\mathcal{E}$ is a completely positive operator, its extension is well defined, and we can let $\chi=(\identity\otimes\mathcal{E})(\proj{\Psi})$. The trace preserving property of $\mathcal{E}$ imposes that $\Tr(\chi)=\Tr\proj{\Psi}=1$ (if it weren't trace preserving, the trace would be non-increasing, $\Tr(\chi)\leq 1$, which doesn't change our conclusion). So, $F=\Tr(R\chi)\leq \lambda\Tr(\chi)\leq \lambda$ where $\lambda$ is the maximum eigenvalue of $R$, and $\chi$ is the corresponding (normalized) maximum eigenvector.
}\end{proof}

The above proof does not guarantee that a map described by state $\chi$ can be implemented; that is the purpose of the following Lemma.

\begin{lemma} \label{lem:achieve}
The upper bound for the achievable fidelity in the state transformation task can be realised if there exists a mixture $\rho$ of the maximum eigenvectors of $R$ such that $\Tr_{\text{OUT}}(\rho)=\identity/d'$, where $\identity$ is over the subspace spanned by the states $\{\ket{\psi_i}\}$, of dimension $d'$. If $\rho$ is a pure state, then the transformation can be achieved {\em economically}.
\end{lemma}
\begin{proof}{
Let the maximum eigenvector of $R$ be $\ket{\chi}$. For any choice of $\ket{\chi}$ (allowing for the fact that the maximum eigenvalue may be degenerate), it can be expressed as a pure bipartite state between the subsystems IN and OUT with a Schmidt decomposition
$$
\ket{\chi}=\sum_{n=0}^{d-1}\beta_n\ket{\eta_n}\ket{\lambda_n},
$$
where $\{\ket{\eta_n}\}$ define an orthonormal basis over the $d'$ dimensional Hilbert space. If there exists a maximum eigenvector such that $\beta_n^2=\smallfrac{d'}$, then
$$
d'\Tr_{\text{IN}}(\proj{\psi}^T\otimes\identity \proj{\chi})=U\proj{\psi}U^\dagger
$$
where
$$
U\ket{\eta_n}=\ket{\lambda_n} \qquad\forall n\in [d'].
$$
Here the relevant Hilbert spaces are extended as necessary so that they have the same size. In this instance, the optimal strategy (application of $U$ to the input state) is called economical, meaning that one does not require an ancilla for the operation to be implemented.

If this cannot be done, but there exists a mixture of maximum eigenvectors such that $\Tr_{\text{OUT}}(\rho)=\identity/d'$, then it is always possible to introduce an ancillary system that purifies $\rho$ and gives Schmidt coefficients between the system IN and the rest of value $1/d'$. By the previous argument, we can therefore implement a unitary operation over this extended space (meaning it is no longer economical) that realizes the desired map.
}\end{proof}

\subsection{Cloning} \label{sec:clone}

We study the $M\rightarrow N$ universal cloning problem. This means that we are given $M$ copies of an unknown pure state, and are tasked with making $N>M$ copies of it, as well as we can. Universal cloning imposes that the unknown quantum state is drawn uniformly from all possible pure states (i.e.\ $U\ket{0}$ where $U$ is drawn uniformly over $SU(d)$). Most studies concentrate on symmetric cloning, in which we want all of the copies produced to be as good as each other. Here, we aim for the loftier goal of wanting to know how we can trade the qualities of the different outputs. Traditionally, one concentrates on the single-clone fidelity:
$$
M(\psi)=\sum_{n=1}^N\alpha_n\identity^{\otimes(n-1)}\otimes\proj{\psi}\otimes\identity^{\otimes(N-n)},
$$
where $\alpha_n\geq 0$ ensures positivity, and $\sum_n\alpha_n=1$ ensures that the maximum value is 1, which only happens for the state $\ket{\psi}^{\otimes N}$. This lets us examine the individual copies. However, other measures have been considered, such as the global fidelity:
$$
M(\psi)=\proj{\psi}^{\otimes N}.
$$
We aim to consider a fully general case where we evaluate the fidelities on arbitrary subsets of qudits. This will be specified by $\Lambda\subseteq\{0,1\}^N$, meaning that any $x\in\Lambda$ wants us to evaluate the fidelity across all sites $n$ for which $x_n=1$, and not over the sites $x_n=0$. As such,
$$
M(\psi)=\sum_{x\in\Lambda}\alpha_x\bigotimes_{n=1}^N\proj{\psi}^{x_n}
$$
where $\sum_{x\in\Lambda}\alpha_x=1$. The shorthand of $\proj{\psi}_x\otimes\identity_{\bar x}$ describes the tensor product $\bigotimes_{n=1}^N\proj{\psi}^{x_n}$  (and by $\proj{\psi}^0$ we understand $\identity$). According to Lemma \ref{lem:transformation}, our task is to find the maximum eigenvalue $\lambda$ and eigenvector $\ket{\chi}$ of the matrix
$$
R=\sum_{x\in\{0,1\}^N}\alpha_xR_x
$$
where
\begin{equation*}
\begin{split}
R_x=&d'\int \proj{\psi}^T_\text{IN}\otimes\proj{\psi}_x\otimes\identity_{\bar x}d\psi	\\
=&d'\int (U\proj{0}U^\dagger)^T_\text{IN}\otimes(U\proj{0}U^\dagger)_x\otimes\identity_{\bar x}dU.
\end{split}
\end{equation*}
The realized fidelity $F=\lambda$ can be described by $\sum_x\alpha_xF_x$ where $F_x$ is the fidelity of the set of clones at the sites $x_n=1$: $F_x=\bra{\chi}R_x\ket{\chi}$. By Lemma \ref{lem:twirl},
$$
R_x^{T_{\text{IN}}}=\frac{\binom{M+d-1}{M}}{\binom{M+w_x+d-1}{M+w_x}}P_{\text{sym}}^{M+w_x,d}.
$$
Lemma \ref{lem:commutator} also shows that $R$ simultaneously commutes with both $J^2$ and $J_Z$ (up to a unitary rotation), and it therefore has two quantum numbers $S$ (where $4S(S+1)$ is an eigenvalue of $J^2$) and $M_Z$ (taking values $-S$ to $S$ in integer steps) that distinguish subspaces of eigenvectors. 

\begin{definition}
We denote by $\ket{\psi_x}$ for $x\in\{0,1\}^N$ and $w_x=M$ the state $\ket{B_x^{(M)}}\ket{\Phi}_{\bar x}$, where $\ket{\Phi}$ is a symmetric state of the $N-M$ spins on $\bar x$.
\end{definition}
The subspace of these states will be show to have particular significance in Lemma \ref{lem:subspace}, but we need to determine some important properties first.

\begin{lemma} \label{lem:eta}
Let $x,y\in\{0,1\}^N$ where $w_y=M$. The state
$$
\ket{\eta_{x,y}}=\sqrt{\frac{\binom{M+d-1}{M}\binom{w_x-x\cdot y+d-1}{d-1}}{\binom{w_x+M-x\cdot y+d-1}{d-1}}}\identity_\text{IN}\otimes P_{\text{sym}}^{x\cup y,d}\otimes\identity_{\bar x\cap\bar y}\ket{\psi_y},
$$
is correctly normalized.
\end{lemma}
\begin{proof}{
Rewriting the projection operator as the integral,
$$
P_{\text{sym}}^{x\cup y,d}=\textstyle\binom{w_x+M-x\cdot y+d-1}{d-1}\int (U\proj{0}U^\dagger)^{\otimes (w_x+M-x\cdot y)}dU,
$$
we have that
\begin{multline*}
\braket{\eta_{x,y}}{\eta_{x,y}}=\textstyle{\binom{M+d-1}{M}\binom{w_x-x\cdot y+d-1}{d-1}}\times\\
\Tr\left(\int \proj{B_y^{(M)}}\otimes\proj{\Phi}_{\bar y}\cdot \identity_{\text{IN}}\otimes(U\proj{0}U^\dagger)^{\otimes (w_x+M-x\cdot y)}\otimes\identity_{\bar{x}\cap\bar y}  dU\right).
\end{multline*}
However, for the spins where $y_n=1$, we can absorb the $U$s into the $\ket{B^{(M)}_y}$ ($U_y$ acts as a unitary within the symmetric subspace of spins $y_n=1$): $\identity_{\text{IN}}\otimes U_y\ket{B^{(M)}_y}=U^{\star}_{\text{IN}}\otimes\identity_y\ket{B^{(M)}_y}$ by Lemma \ref{lem:transpose}, and they cancel, leaving
\begin{eqnarray*}
\braket{\eta_{x,y}}{\eta_{x,y}}&=&\textstyle{\binom{M+d-1}{M}\binom{w_x-x\cdot y+d-1}{d-1}}\times \\
&&\Tr\left(\int \proj{B_y^{(M)}}\otimes\proj{\Phi}\cdot \identity_{\text{IN}}\otimes\proj{0}_y\otimes(U\proj{0}U^\dagger)_{\otimes x\cap\bar y}\otimes\identity_{\bar{x}\cap\bar y}  dU\right)	\\
&=&\textstyle\binom{w_x-x\cdot y+d-1}{d-1}\Tr\left(\int \proj{\Phi}\cdot (U\proj{0}U^\dagger)_{x\cap\bar y}\otimes\identity_{\bar{x}\cap\bar y}  dU\right)\\
&=&\Tr\left(\proj{\Phi}\cdot P_{\text{sym}}^{x\cap\bar y,d}\otimes\identity_{\bar{x}\cap\bar y} \right).
\end{eqnarray*}
Since $\ket{\Phi}$ is a $+1$ eigenstate of all possible permutations of its spins, this includes the permutations involved on the subset of spins $x\cap\bar y$. Thus, the trace has value 1 and we are left with
$\braket{\eta_{x,y}}{\eta_{x,y}}=1$.
}\end{proof}

\begin{lemma}\label{lem:inner}
For any two binary strings $x,y\in\{0,1\}^N$ with $w_x=w_y=M$,
$$
\braket{\psi_x}{\psi_y}=\frac{1}{\binom{M-x\cdot y+d-1}{d-1}}.
$$
\end{lemma}
\begin{proof}{
Clearly, the value $\braket{\psi_x}{\psi_y}$ will only depend on the value $x\cdot y$, and not on the specific choices of $x$ and $y$. We will prove the value by induction. As a base case, take $x\cdot y=M$, i.e. $x=y$. Evidently, $\braket{\psi_x}{\psi_y}=1$, as predicted. For the inductive step, assume this formula holds for all values of $x\cdot y=k+1,\ldots, M$, and we aim to show that it holds for $x\cdot y=k$.

Select $x$ and $y$ such that $x\cdot y=k$. Now consider the normalised state $\ket{\eta_{x,y}}$ of Lemma \ref{lem:eta}:
$$
\ket{\eta_{x,y}}=\sqrt{\frac{\binom{M+d-1}{M}\binom{M-k+d-1}{d-1}}{\binom{2M-k+d-1}{d-1}}}\frac{1}{\binom{2M-k}{M}}\sum_{\substack{z\in\{0,1\}^N\\ w_z=M\\ z\cdot(x\cup y)=M}}\ket{\psi_z}.
$$
Taking the inner product gives
$$
\frac{\binom{2M-k}{M}\binom{2M-k+d-1}{d-1}}{\binom{M+d-1}{M}\binom{M-k+d-1}{d-1}}={\textstyle\binom{M}{k}}\braket{\psi_x}{\psi_y}+\sum_{q=k+1}^M\frac{\binom{M}{q}\binom{M-k}{M-q}}{\binom{M-q+d-1}{d-1}},
$$
making use of the inductive assumption to give the denominator of the final term. Hence,
$$
\frac{\binom{2M-k+d-1}{M}}{\binom{M+d-1}{M}}={\textstyle\binom{M}{k}}\left(\braket{\psi_x}{\psi_y}-\frac{1}{\binom{M-k+d-1}{d-1}}\right)+\frac{\binom{2M-k+d-1}{M}}{\binom{M+d-1}{M}}.
$$
Rearranging gives the desired result:
$$
\braket{\psi_x}{\psi_y}=\frac{1}{\binom{M-k+d-1}{d-1}}.
$$
}\end{proof}

\begin{lemma}\label{lem:subspace}
$\mathcal{S}_{\text{special}}:=\text{span}\{\ket{\psi_y}\}$ is a closed subspace under the action of $\{R_x\}$. By fixing a $J_Z$ subspace for the symmetric state $\ket{\Phi}$, $\mathcal{S}_{\text{special}}$ is a subspace of fixed quantum number $M_Z$.
\end{lemma}
\begin{proof}{
The previous Lemma conveyed that $\ket{\eta_{x,y}}$ is supported on $\mathcal{S}_{\text{special}}$. We now claim that $\ket{\tilde\eta_{x,y}}=\ket{\eta_{x,y}}$ where
$$
\ket{\tilde\eta_{x,y}}:=\frac{\sqrt{\binom{M+d-1}{M}\binom{w_x-x\cdot y+d-1}{d-1}\binom{w_x+M-x\cdot y+d-1}{d-1}}}{\binom{M+w_x+d-1}{d-1}}{P_{\text{sym}}^{(\text{IN},x),d}}^{T_{\text{IN}}}\ket{\psi_y},
$$
i.e.\ we aim to show that $\braket{\tilde\eta_{x,y}}{\tilde\eta_{x,y}}=1$ and $\braket{\tilde\eta_{x,y}}{\eta_{x,y}}=1$.
\begin{multline*}
\frac{\braket{\tilde\eta_{x,y}}{\tilde\eta_{x,y}}}{\binom{w_x-x\cdot y+d-1}{d-1}\binom{w_x+M-x\cdot y+d-1}{d-1}\binom{M+d-1}{M}}=\\\!\!\iint\!\!\bra{\psi_y}(U^*\proj{0}U^TV^*\proj{0}V^T)_{\text{IN}}\otimes (U\proj{0}U^\dagger V\proj{0}V^\dagger)_x\otimes\identity_{\bar x}\ket{\psi_y}dUdV
\end{multline*}
Moving the unitaries through the states $\ket{\Psi}$ where possible, this reduces to
\begin{multline*}
\frac{\braket{\tilde\eta_{x,y}}{\tilde\eta_{x,y}}}{\binom{w_x-x\cdot y+d-1}{d-1}\binom{w_x+M-x\cdot y+d-1}{d-1}}=\\\!\!\iint\!\!\bra{\Phi}_{\bar y}(U\proj{0}U^\dagger V\proj{0}V^\dagger)_{x\cap \bar y}\otimes\identity_{\bar x\cap \bar y}\ket{\Phi}_{\bar y}\bra{0}U^\dagger V\ket{0}^{M+w_x-x\cdot y}\bra{0}U^TV^*\ket{0}^MdUdV
\end{multline*}
We substitute $W=U^\dagger V$ (eliminating $V$) and integrate over $U$, giving
\begin{eqnarray*}
\braket{\tilde\eta_{x,y}}{\tilde\eta_{x,y}}&=&\textstyle\binom{w_x+M-x\cdot y+d-1}{d-1}\int|\bra{0}W\ket{0}|^{2(M+w_x-x\cdot y)}dW\\&=&1
\end{eqnarray*}
Calculation of $\braket{\eta_{x,y}}{\tilde\eta_{x,y}}$ follows an identical trajectory, yielding the desired result.
}\end{proof}

We therefore see that there are many degenerate spaces (different choices of the state $\ket{\Phi}$), and eigenvectors within these subspaces can be described by states
$$
\ket{\chi}=\sum_{\substack{x\in\{0,1\}^N\\ w_x=M}}\beta_x\ket{\psi_x}
$$
for coefficients $\{\beta_x\}$ which are normalized as
\begin{equation}
\sum_{x,z}\frac{\beta_x\beta_z}{\binom{M-x\cdot z+d-1}{d-1}}= 1\quad\Leftrightarrow\quad \underline{\beta}^TG_{\underline 0}^{(M)}\underline{\beta}= 1		\label{eqn:norm}
\end{equation}
by Lemma \ref{lem:inner}. Here we have grouped all the parameters $\beta_x$ into a single vector $\underline{\beta}$. Each of the individual fidelities $F_y=\bra{\chi}R_y\ket{\chi}$ can be evaluated as
$$
\bra{\chi}R_y\ket{\chi}={\textstyle{\binom{M+d-1}{M}}}\!\!\sum_{\substack{x,z\in\{0,1\}^N\\ w_x=w_z=M}}\!\!\beta_x\beta_z\frac{\bra{\psi_x}{P_{\text{sym}}^{(\text{IN},y),d}}^{T_{\text{IN}}}\ket{\psi_z}}{\binom{M+w_y+d-1}{d-1}}
$$
By Lemma \ref{lem:subspace}, this simplifies to
$$
F_y=\underline{\beta}^TG_y^{(M)}\underline{\beta}.
$$
Note that in the next section we will prove that the maximum eigenvector will correspond to all the entries $\beta_x\geq 0$, so we can use $\underline{\beta}^T$ instead of $\underline{\beta}^\dagger$ without loss of generality.

\subsubsection{Lieb-Mattis Theorem} \label{sec:liebmattis}

So far, we have demonstrated that the maximum eigenvalue of the matrix $R$ is an upper bound to the optimal cloning fidelity. We have studied a particular subspace of $R$, $\mathcal{S}_{\text{special}}$, so it remains to prove that this subspace contains the maximum eigenvector of $R$, and that this fidelity can be achieved. We do this by modifying the Lieb-Mattis Theorem \cite{lieb1962}.

\begin{lemma}\label{lem:liebmattis}
The maximum eigenvector of $R$ within a given $M_Z$ subspace has non-negative coefficients on all the basis states.
\end{lemma}
\begin{proof}{
Divide $R$ into two components, the diagonal elements ($R_d$) and the remaining, off-diagonal, elements ($R_o$). In the computational basis, $\ket{a}$, we have $\bra{a}R_d\ket{a}=e_a$ and $\bra{a}R_o\ket{b}=K_{ab}$. Note that $K_{ab}\geq 0$ by Corollary \ref{lem:positive} and $K_{aa}=0$ by definition. Assume that in a particular excitation subspace, $M_Z$, we know the maximum eigenvector,
$$
\ket{\chi}=\sum_a f_a\ket{a},
$$
with eigenvalue $E_{M_Z}$. Hence,
$$
\sum_b K_{ab}f_b=(E_{M_Z}-e_a)f_a.
$$
Any other state must have a smaller expectation value of $R$, unless it is also a maximum eigenvector. Let us first try a state $\ket{a}$ as the ansatz. This reveals $e_a\leq E_{M_Z}$. Hence, we can take the modulus of the above equation,
$$
\left|\sum_bK_{ab}f_b\right|=(E_{M_Z}-e_a)|f_a|.
$$
Next, consideration of the ansatz state
$$
\ket{\tilde\chi}=\sum_a|f_a|\ket{a}
$$
imposes that there is at least one non-zero $f_a$ such that
$$
\sum_bK_{ab}|f_b|\leq(E_{M_Z}-e_a)|f_a|
$$
but since $\sum_bK_{ab}|f_b|\geq\left|\sum_bK_{ab}f_b\right|$, this can only be satisfied with equality for every $a$, meaning that, up to a global phase factor, the coefficients of the maximum eigenvector in each $M_Z$ subspace satisfy
$
f_a\geq 0.
$ 
}\end{proof}

\begin{corollary} \label{cor:important}
The maximum eigenvector of $R$ is contained within the span of states of $\mathcal{S}_{\text{special}}$.
\end{corollary}
\begin{proof}{
Consider the state $\ket{\chi_{\text{sym}}}$ for which all the $\beta_x$ are chosen to be equal, and the state $\ket{\Phi}$ is chosen to be the uniform superposition of all basis states of $N-M$ spins within a fixed $M_Z$ subspace (total excitation number). Overall, this state contains all basis states $\ket{a}\ket{b}$ for $a\in\{0,1\}^M$ and $b\in\{0,1\}^N$ provided $a$ exists as a subset of $b$, for a fixed value of $M_Z=M_Z(b)-M_Z(a)$, and has a positive amplitude on all such states. Note that the remaining basis states cannot contribute to the optimal cloner: the output of the cloning map, $\ket{b}$, when applied to an input $\ket{a}$, would be orthogonal to the input for all possible clones. We can therefore discount these states (intuitively, if a subset of $a$ of size $\tilde M<M$ can be found as a subset of $b$, this corresponds to $\tilde M\rightarrow N$ cloning, which must be worse than $M\rightarrow N$ cloning). From Lemma \ref{lem:liebmattis}, the maximum eigenvector $\ket{\chi}$ for a given $M_Z$ has $\braket{\chi}{\chi_{\text{sym}}}>0$. Since $\ket{\chi_{\text{sym}}}$ is in $\mathcal{S}_{\text{special}}$, the only way that this can happen is if $\ket{\chi}\in\mathcal{S}_{\text{special}}$.

This establishes that for the subspaces $M_Z=-\half(d-1)(N-M),\ldots ,\half(d-1)(N-M)$, the maximum eigenvector is found in $\mathcal{S}_{\text{special}}$. The maximum eigenvalue in the other subspaces cannot be larger -- the eigenvector corresponding to the larger eigenvalue would have to be drawn from a subspace of $S>\half(d-1)(N-M)$. However, any such eigenvalue is degenerate in $M_Z$ so it would also be present in all other $M_Z$ subspaces from $-S$ to $S$, in particular in the ones in which we have already found a different maximum eigenvector.
}\end{proof}

In principle, we now have an upper bound on the cloning fidelity. Can this be achieved? We invoke Lemma \ref{lem:achieve}, where, although the coefficients $\{\beta_x\}$ are fixed, we are free to pick $\ket{\Phi}$ to be an arbitrary symmetric state, or a mixture thereof.

\begin{lemma}\label{lem:upper}
For the upper bound on the cloning fidelity to be achievable, it suffices to find a symmetric state $\ket{\Phi}$ of $N-M$ spin $d$ systems, or mixture thereof, such that
$$
\sum_{\Phi}p_{\Phi}\Tr_{\min(M+1,N-M+1),\ldots,N-M}\proj{\Phi}=\frac{P_{\text{sym}}^{\min(M,N-M),d}}{\binom{\min(M,N-M)+d-1}{d-1}}.
$$
\end{lemma}
\begin{proof}{
If the upper bound is to be achieved, then Lemma \ref{lem:achieve} requires that there exists $\{(p_{\Phi},\ket{\chi(\Phi)})\}$ satisfying $\sum_{\Phi}p_{\Phi}=1$ such that
$$
\sum_{\Phi}p_{\Phi}\Tr_{\text{OUT}}\proj{\chi(\Phi)}=\frac{1}{\binom{M+d-1}{M}}P^{M,d}_{\text{sym}}.
$$
This is equivalent to
\begin{equation}
\sum p_\Phi\braket{\lambda_i(\Phi)}{\lambda_j(\Phi)}=\delta_{i,j}		\label{eqn:quickref}	
\end{equation}
where
$$
\ket{\lambda_i(\Phi)}=\sum_x\beta_x\ket{\phi_i^x}\ket{\Phi}_{\bar x}.
$$
Recalling from Lemma \ref{lem:liebmattis} that all the matrix elements must be non-negative (which includes the $\{\beta_x\}$), this reduces the condition on the off-diagonal elements to
$$
\bra{\phi_i^{(M)}}\Tr_{N-2M+x\cdot z}(\proj{\Phi})\otimes\identity_{x\cdot z}\ket{\phi_j^{(M)}}=0\qquad \forall i\neq j,\forall x,z\in\{0,1\}^N:w_x=w_z=M.
$$
We have dropped an explicit enumeration of which spins the states apply to because this purely mathematical manipulation has removed such a direct connection, and the partial trace is thus just taken over any set of $N-2M+x\cdot z$ spins. This must be true for all values of $x\cdot z=\max(0,2M-N),\ldots,M$ and for all $i=1,\ldots\binom{M+d-1}{M}$, which in turn imposes that
$\Tr_{N-M-k}\left(\sum_{\Phi} p_\Phi\proj{\Phi}\right)$
is diagonal in the symmetric subspace for $k=0,1,\ldots \min(M,N-M)$. Now, observe that if we set $\rho_{\Phi}=\sum_{\Phi}p_{\Phi}\proj{\Phi}$, and
$$\Tr_{N-M-\min(M,N-M)}\left(\rho_{\Phi}\right)=\frac{P_{\text{sym}}^{\min(M,N-M),d}}{\binom{\min(M,N-M)+d-1}{d-1}},$$
then
$$\Tr_{N-M-k}\left(\rho_{\Phi}\right)=\frac{P_{\text{sym}}^{k,d}}{\binom{k+d-1}{k}}\qquad\forall k$$
and we also get the correct value for the diagonal elements, $i=j$ in Eq.\ (\ref{eqn:quickref}):
\begin{eqnarray*}
\sum p_\Phi\braket{\lambda_i(\Phi)}{\lambda_i(\Phi)}&=&\sum_{k=0}^{\min(M,N-M)}\bra{\phi_i^{(M)}}\Tr_{N-M-k}(\rho_{\Phi})\otimes\identity_{M-k}\ket{\phi_i^{(M)}}\sum_{x,z:x\cdot z=M-k}\beta_x\beta_z	\\
&=&\sum_{k=0}^{\min(M,N-M)}\bra{\phi_i^{(M)}}\frac{P_{\text{sym}}^{k,d}}{\binom{k+d-1}{k}}\otimes\identity_{M-k}\ket{\phi_i^{(M)}}\sum_{x,z:x\cdot z=M-k}\beta_x\beta_z	\\
&=&\sum_{k=0}^{\min(M,N-M)}\frac{\braket{\phi_i^{(M)}}{\phi_i^{(M)}}}{\binom{k+d-1}{k}}\sum_{x,z:x\cdot z=M-k}\beta_x\beta_z	\\
&=&\underline{\beta}^T\cdot G_{\underline{0}}^{(M)}\cdot\underline{\beta}	\\
&=& 1
\end{eqnarray*}
by virtue of the normalisation condition, Eq.\ (\ref{eqn:norm}). 
}\end{proof}

\begin{corollary}\label{cor:1}
The optimal cloning fidelity can always be achieved.
\end{corollary}
\begin{proof}{
We can pick a mixture of symmetric states
$$
\rho_{\Phi}=\frac{P_\text{sym}^{N-M,d}}{\binom{N-M+d-1}{d-1}},
$$
which clearly satisfies the condition of Lemma \ref{lem:upper}, so we can certainly always realise the optimal cloning transformation. However, this version requires the use of an ancilla spin (or subspace) of dimension  $\binom{N-M+d-1}{d-1}$ if one wants to describe a pure state. For example, in place of $\ket{\Phi_{\bar x}}$, one uses a maximally entangled state (within the symmetric subspace) between the $\bar x$ spins of the out space and the equivalent number ($N-M$) ancilla spins.
}\end{proof}

This concludes the proof of Thm. \ref{thm:main}; we have proven an upper bound on the achievable fidelity, and shown that this bound can always be obtained. We are now in the position that, given a set $\{\alpha_x\}$, we can find the corresponding optimal cloning fidelities, simply by solving for the maximum eigenvector of an $\binom{N}{M}\times\binom{N}{M}$ matrix. However, being given the set $\alpha_x$ is an unnatural setting, and was merely a mathematical convenience. It would be far more useful to be able to either describe the region of achievable cloning fidelities and how they trade-off against each other, or to be able to ascertain whether a given set of fidelities are achievable.

The former question makes most sense when the possible fidelities are constrained. For example, if we consider the set of fidelities $\Lambda=\{y\in\{0,1\}^N:w_y=L\}$, then instead of trying to find the maximum eigenvector as a function of $\{\alpha_x\}$ and eliminating them, we use the fidelities to eliminate the $\{\beta_x\}$ from the normalization condition.

\subsection{The Economy of Cloning}
Economical cloning means that the cloning transformation can be realised without the use of ancillas. In the previous section, we saw how the optimal cloning transformation could be realised, but the construction required $N-M$ ancilla qudits, when using their symmetric subspace, or, equivalently, a single ancilla of dimension $\binom{N-M+d-1}{d-1}$. Under what circumstances can the ancilla be removed?

\begin{lemma}\label{cor:2}
Economical cloning is impossible if $N< 3M$.
\end{lemma}
\begin{proof}{
For economy, we demand the existence of a pure state $\ket{\Phi}$. However, $\rho_{\Phi}$ must satisfy
$$
\Tr_{\min(M+1,N-M+1),\ldots,N-M}\rho_{\Phi}=\frac{P_{\text{sym}}^{\min(M,N-M),d}}{\binom{\min(M,N-M)+d-1}{d-1}}.
$$
Lemma \ref{lem:upper} proved this was sufficient, but here we claim necessity. We already know that $\Tr_{\min(M+1,N-M+1),\ldots,N-M}\rho_{\Phi}$ must be diagonal in the symmetric subspace, so let
$$\Tr_{\min(M+1,N-M+1),\ldots,N-M}\rho_{\Phi}=\sum_{n=0}^{\binom{\min(M,N-M)+d-1}{d-1}}\eta_n\proj{\phi^{(\min(M+1,N-M+1))}_n}.$$
Given that we require $\sum_{\Phi}p_{\Phi}\braket{\lambda_i(\Phi)}{\lambda_i(\Phi)}=1$ for all $i$, this gives $\binom{M+d-1}{M}$ linear conditions in the $\eta_n$ which must span the entire space, i.e.\ there can be no more than one solution for the values of the $\eta_n$. However, we already know one solution, where all the $\eta_n$ are equal, so this must be the only solution. For example, if $d=2$ and $M=3,N\geq 6$, one has
$$
\left(\begin{array}{cccc}
B_1+B_2+B_3 & \frac{2}{3}B_1+\frac{B_2}{3} & \frac{B_1}{3} & 0	\\
\frac{2}{3}B_1+\frac{B_2}{3} & \frac{5}{9}(B_1+B_2) & \frac{4}{9}(B_1+B_2) & \frac{B_1}{3} \\
\frac{B_1}{3} & \frac{4}{9}(B_1+B_2) & \frac{5}{9}(B_1+B_2) & \frac{2}{3}B_1+\frac{B_2}{3} \\
0 & \frac{B_1}{3} & \frac{2}{3}B_1+\frac{B_2}{3} & B_1+B_2+B_3 
\end{array}\right)\left(\begin{array}{c} \eta_0 \\ \eta_1 \\ \eta_2 \\ \eta_3 \end{array}\right)=(1-B_0)\left(\begin{array}{c} 1 \\ 1 \\ 1 \\ 1 \end{array}\right)
$$
where $B_k=\sum_{x,z:x\cdot z=3-k}\beta_x\beta_z$ and $4B_0+2B_1+\frac{4}{3}B_2+B_3=4$ is equivalent to the statement $\underline{\beta}^T\cdot G_{\underline{0}}^{(M)}\cdot\underline{\beta}$. Here, each of the 4 rows corresponds to a different condition $\bra{\phi_i^{(M)}}\Tr_{4,\ldots,N-M}(\rho_{\Phi})\otimes\identity\ket{\phi_i^{(M)}}=1$, $i=0,\ldots,3$. The only solution is $\eta_k=\frac{1}{4}$.

Now, consider the Schmidt decomposition of any such state when the spins of state $\ket{\Phi}$ are split into a bipartition of $\min(M,N-M)$ vs.\ $\max(N-2M,0)$. This means that the dimension of the symmetric subspace of the remaining $\max(N-2M,0)$ spins must be at least $\binom{\min(M,N-M)+d-1}{d-1}$, i.e.\ $N\geq 3M$ (since $N>M$).
}\end{proof}

A trivial example of the economy of cloning is the case of $M=1$. We now know that an ancilla, often known as an `anti-clone' \cite{durt2005}, is required for $N=2$. For $N\geq 3$, we have to find a symmetric state that has partial trace on a single spin of the maximally mixed state:
$$
\ket{\Phi}=\frac{1}{\sqrt{d}}\sum_{i=0}^{d-1}\ket{i}^{\otimes(N-M)}.
$$
Such choices are not unique \cite{baguette2014}. Nevertheless, it seems that economical cloning is an unusual property. For $d=2$, we can construct states for $M=2,3$ and $N=3M$. For instance,
$$
\sqrt{3}\ket{\Phi}=\ket{00}\frac{\ket{01}+\ket{10}}{\sqrt{2}}+\frac{\ket{01}+\ket{10}}{\sqrt{2}}\ket{00}+\ket{11}\ket{11}
$$
is permutation invariant and has the requisite Schmidt basis for $M=2$. Meanwhile,
\begin{lemma} \label{lem:econ2}
For $d=2$ and $M\geq 4$, economical universal cloning is impossible.
\end{lemma}
\begin{proof}{
It is sufficient to prove that for $M=4$ and $N\geq 12$ economical universal cloning is impossible, because if the partial trace of $\proj{\Phi}$ onto just four qubits does not give the projector onto the symmetric subspace, then the partial trace onto any larger number of qubits cannot give a projector onto the symmetric subspace (because that projector's partial trace would, itself, be a projector onto the symmetric subspace). We write
$$
\ket{\Phi}=\sum_{i=0}^{N-4}\alpha_i\ket{\phi^{(N-4)}_i}
$$
where the states $\ket{\phi^{(N-4)}_i}$ are the uniform superpositions of $i$ $\ket{1}$ states (a basis of the symmetric subspace of $N-4$ qubits). We require that $\sum_i|\alpha_i|^2=1$ and that $\alpha_i\geq 0$ for all $i$ (since the maximum eigenvector must have non-negative coefficients). Now examine the Schmidt decomposition of the states $\ket{\phi_i^{(N-4)}}$:
$$
\ket{\phi_i^{(N-4)}}=\sum_{j=\max(0,i+8-N)}^{\min(i,4)}\frac{\binom{4}{j}\binom{N-8}{i-j}}{\binom{N-4}{i}}\ket{\phi^{(4)}_j}\ket{\phi^{(N-8)}_{i-j}}.
$$
So, if two terms $\alpha_i$ and $\alpha_k$ are non-zero ($i\neq k$), they yield an off-diagonal term in $\Tr_{N-8}\proj{\Phi}$, as written in the basis of the symmetric states, whenever there exist $j$ and $l$ such that $i-j=k-l$ within their appropriate summation ranges. If such a term arises, at least one of $\alpha_i,\alpha_k$ must be 0 for the state to be diagonal in the symmetric basis. In particular, for $i=4,\ldots,N-8$, this imposes that $\alpha_i=0$ because these states $\ket{\phi_i}$ give off-diagonal terms with every other possible state. Thus, to ensure a diagonal outcome, we either pick $\alpha_i=0$ or all $\alpha_k=0$ for $k\neq i$. In the latter case, it is easy to verify that $\Tr_{N-8}\proj{\phi_i}\neq P_{\text{sym}}^4$. A similar analysis continues between all the $i=0,1,2,3$, where we conclude that only one of them can be non-zero to ensure that the output is diagonal. And, again, for all the $i=N-8+1,\ldots N-4$. So, let us pick $i\in\{0,1,2,3\}$ and $j\in\{N-7,N-6,N-5,N-4\}$ to be $\alpha_i$ and $\alpha_j$ are the only non-zero terms. What are the matrix elements of the $\proj{0000}$ and $\proj{1111}$ components of the reduced state of $\proj{\Phi}$?
$$
|\alpha_i|^2\frac{\binom{N-8}{i}}{\binom{N-4}{i}}	\qquad |\alpha_j|^2\frac{\binom{N-8}{j-4}}{\binom{N-4}{j}},
$$
both of which need to be $1/5$ in order to get the projector on the symmetric subspace. However, we also require the normalization condition for $\ket{\Phi}$:
$$
|\alpha_i|^2+|\alpha_j|^2=\frac{1}{5}\left(\frac{\binom{N-4}{i}}{\binom{N-8}{i}}+\frac{\binom{N-4}{j}}{\binom{N-8}{j-4}}\right)=1.
$$
By iterating through all possible values of $i,j$ and solving the equation for $N$, we find that there is never an integer value of $N$ that is a valid solution.
}\end{proof}

\section{Linear Constraints on Cloning}\label{sec:linear}

The contents of Thm.\ \ref{thm:constraints} describe the reduction of the quadratic constraints for specifying fidelities down to linear ones, which reduce the complexity of solution of the system. These are inspired by the derivation of fidelity relations if $M=1$ and $L=1$ or $N-1$ in Sec.\ \ref{sec:1}.

\subsection{\texorpdfstring{$1\rightarrow N$}{1 to N} cloning} \label{sec:1p2}

The essential feature of these derivations of optimal cloning fidelities in Sec.\ \ref{sec:1} was the implicit ability to find linear combinations of the matrices $G_y^{(M)}$ and $G_{\underline 0}^{(M)}$ that are rank 1. This lets us reduce the quadratic constraints $\underline{\beta}^TG_y\underline{\beta}=F_y$ to linear ones on the $\{\beta_x\}$: if there exists a set of coefficients $\{g_y\}$ such that
$$
g_0G_{\underline 0}+\sum g_yG_y=\underline{\Gamma}\cdot\underline{\Gamma}^T
$$
where $\underline{\Gamma}$ is a unit vector, then
$$
\underline{\Gamma}^T\underline{\beta}=\sqrt{g_0+\sum g_yF_y}.
$$

\begin{lemma} \label{lem:linear}
In $(1,L,N)$ cloning, there exist values $g_0,g_1,g_2\in\mathbb{R}$ that yield a rank 1 projector
$$
P_1=g_0G_{\underline 0}+g_1G_{(1,0,0,\ldots,0)}^{(M,L)}+g_2G_{\underline 0}^{(M,L)}.
$$
\end{lemma}
\begin{proof}{
We observe that each of the 3 matrices $G_{\underline 0}, G_{(1,0,0,\ldots,0)}^{(M,L)}, G_{\underline 0}^{(M,L)}$ has the form
\begin{eqnarray}
G_{\underline 0}&=&\frac{1}{d}(\ket{1}+\sqrt{N-1}\ket{j})(\bra{1}+\sqrt{N-1}\bra{j})+\frac{d-1}{d}\identity	\nonumber\\
G_{(1,0,0,\ldots,0)}^{(M,L)}&=&a_1\proj{1}+a_2(\ket{1}\bra{j}+\ket{j}\bra{1})+a_3\proj{j}+a_0\identity		\label{eqn:form}\\
G_{\underline 0}^{(M,L)}&=&a_4(\ket{1}+\sqrt{N-1}\ket{j})(\bra{1}+\sqrt{N-1}\bra{j})+a_5\identity	\nonumber
\end{eqnarray}
where $\ket{j}=\sum_{n=2}^{N}\ket{n}/\sqrt{N-1}$ and we are using an index $n\in[N]$ in place of a bit string of length $N$ with a single entry 1 at position $n$, and
\begin{eqnarray*}
a_0&=&a_2-a_3	\\
a_1&=&\frac{\binom{N-1}{L-1}}{\binom{d+L-2}{d-1}}-a_0		\\
a_2&=&\frac{\binom{N-1}{L-1}}{N-1}\left(\frac{L-1}{\binom{d+L-2}{d-1}}+\frac{N-L}{\binom{d+L-1}{d-1}}\right)		\\
a_3&=&\frac{\binom{N-1}{L-1}}{\binom{N-1}{2}}\left(\frac{\binom{L-1}{2}}{\binom{d+L-2}{d-1}}+\frac{(L-1)(N-L)}{\binom{d+L-1}{d-1}}+\frac{\binom{N-L}{2}}{\binom{d+L}{d-1}}\right)		\\
a_4&=&\frac{\binom{N}{L}}{\binom{N}{2}}\left(\frac{\binom{L}{2}}{\binom{d+L-2}{d-1}}+\frac{L(N-L)}{\binom{d+L-1}{d-1}}+\frac{\binom{N-L}{2}}{\binom{d+L}{d-1}}\right)		\\
a_5&=&\frac{\binom{N}{L}}{N}\left(\frac{L}{\binom{d+L-2}{d-1}}+\frac{N-L}{\binom{d+L-1}{d-1}}\right)-a_4.	\\
\end{eqnarray*} 
Any linear combination also has the form of Eq.\ (\ref{eqn:form}), such that by judicious choice of $\{g_i\}$, a rank 1 projector of the form $\underline{\Gamma}=\gamma_1\ket{1}+\gamma_2\sqrt{N-1}\ket{j}$ results, where $\gamma_1\neq \gamma_2$ and
\begin{eqnarray*}
\frac{g_0}{d}+g_1a_1+g_2a_4&=&\gamma_1^2	\\
\frac{g_0}{d}+g_1a_3+g_2a_4&=&\gamma_2^2	\\
\frac{d-1}{d}g_0+a_5g_2+g_1a_4&=&0	\\
\gamma_1^2+(N-1)\gamma_2^2&=&1	\\
\left(\frac{g_0}{d}+g_1a_2+g_2a_4\right)^2&=&\gamma_1^2\gamma_2^2.
\end{eqnarray*}
Provided $g_1\neq 0$, the last condition reduces to the linear one
$$
g_0\frac{a_1+a_3-2a_2}{d}+g_1(a_1a_3-a_2^2)+g_2a_4(a_3-a_2)=0.
$$

It can only be the case that $\gamma_1=\gamma_2$ if $a_1=a_2$. This is equivalent to
$$
a_2=\frac{\binom{N-1}{L-1}}{\binom{d+L-2}{d-1}},
$$
which further simplifies to $(N-L)(d-1)=0.$ Hence, provided $L\neq N$ (whose solution is already known), we know that $\gamma_1\neq\gamma_2$.
}\end{proof}

\noindent Consequently, we can write that
\begin{equation}
\gamma_2\sum_{m=1}^N\beta_m+(\gamma_1-\gamma_2)\beta_n=\sqrt{g_0+g_1\sum_{\substack{y\in\Lambda\\y_1=1}}F_y+g_2\sum_{y\in\Lambda}F_y}.	\label{eq:eliminate}
\end{equation}
Given that $\gamma_1\neq\gamma_2$, an equivalent but independent condition can be derived for each of the $N$ sites (singling out a different $\beta_n$). By summing all of these, we get
$$
((N-1)\gamma_2+\gamma_1)\sum_m\beta_m=\sum_{n=1}^N\sqrt{g_0+g_1\sum_{\substack{y\in\Lambda \\y_n=1}}F_y+g_2\sum_{y\in\Lambda}F_y},
$$
and hence we can calculate each of the individual $\beta_n$ in terms of the fidelities. With the $\beta_n$ in place, we simply have to verify if all the cloning fidelities, and the normalization condition, are satisfied. For $L=1,N-1$, there are no outstanding quadratic conditions aside from normalization, and the previous results are recovered.

For $1<L<N-1$, we consider our original problem, which can be phrased as the satisfiability problem
$$
\min_{\substack{\underline{\beta}^TG_{\underline 0}\underline{\beta}=1 \\ \underline{\beta}^TG_y\underline{\beta}=\tilde F_y\quad\forall y\\\tilde F_y\geq F_y\quad\forall y}} 1,
$$
where $F_y$ are the target fidelities and the free parameters are the $\binom{N}{M}$ parameters $\beta_n$ and the $\binom{N}{L}$ parameters $\tilde F_y$. By replacing the $\binom{N}{L}$ quadratic conditions with $N$ linear ones (in $\beta_n$), just derived, the problem becomes convex, and hence efficiently solvable \cite{boyd2004} but misses out some of the constraints -- satisfaction is necessary but not sufficient.

\subsubsection{Consistency Relations} This situation can be improved by verifying the existence of consistency conditions between the fidelities. These will yield further linear relations that can be incorporated into the initial solution, meaning that there will only be $\binom{N}{2}-N$ quadratic constraints left to verify, independent of $L$ (the forthcoming Lemma will return a space of $\binom{N}{2}$ quadratic constraints that need to be verified but Lemma \ref{lem:linear} allows a further $N$ to be removed).
\begin{lemma}\label{lem:consistent}
Define the $\binom{N}{2}\times\binom{N}{L}$ matrix $X$ as $\bra{x}X\ket{y}=\delta_{x\cdot y=2}$ where $x,y\in\{0,1\}^N$ and $w_x=2, w_y=L$. Any $\underline{v}\in\text{Ker}(X)$ satisfies $\sum_{y\in\Lambda}v_yG_y=0$.
\end{lemma}
\begin{proof}{
We need to calculate both the diagonal
$$
\bra{n}\sum_{y\in\Lambda}v_yG_y\ket{n}=\frac{1}{\binom{d+L-1}{d-1}}\sum_{y\in\Lambda}v_y\left(1+y_n\frac{d-1}{L}\right)
$$
and off-diagonal, $\bra{n}\sum_{y\in\Lambda}v_yG_y\ket{m}:$
$$
\frac{1}{\binom{d+L}{d-1}}\sum_{y\in\Lambda}v_y\left(1+(y_n+y_m)\frac{d-1}{L+1}+y_ny_m\frac{d(d-1)}{L(L+1)}\right)
$$
matrix elements. Both are 0 due to the relations
\begin{equation*}
\begin{split}
\frac{1}{\binom{L}{2}}\sum_{x\in\{0,1\}^N:w_x=2}\bra{x}X\ket{v}=\sum_yv_y&=0	\\
\bra{n,m}X\ket{v}=\sum_{y}y_ny_mv_y&=0	\\
\frac{1}{L-1}\sum_{m\neq n}\bra{n,m}X\ket{v}=\sum_yy_nv_y&=0
\end{split}
\end{equation*}
since $\underline{v}$ is in the Null space of $X$, and we have used $\ket{n,m}$ as a synonym for a weight-2 binary string where the 1s are at sites $n\neq m$.
}\end{proof}

\noindent Any optimal set of fidelities $\{F_y\}$, expressed as a vector $\underline{F}$, must satisfy $\underline{v}\cdot\underline{F}=0$ for all $\underline{v}\in\text{Ker}(X)$, which yields $\binom{N}{L}-\binom{N}{2}$ independent conditions. We can therefore formulate our best solution to the cloning problem as a convex optimization problem \cite{boyd2004} to satisfy
\begin{multline*}
d(\gamma_1-\gamma_2)^2+(d-1)\left(Ng_0+\frac{Ng_2}{L}q+g_1q\right)=\\ \frac{\sum_n\sqrt{g_0+g_1F_n+\frac{g_2}{L}q}}{((N-1)\gamma_2+\gamma_1)^2}\left((\gamma_1-\gamma_2)^2-(d-1)\gamma_2(N\gamma_2+2(\gamma_1-\gamma_2))\right)
\end{multline*}
(the normalisation condition) subject to the constraints
\begin{eqnarray*}
F_n&=&\sum_{y:y_n=1}\tilde F_y \\
q&=&\sum_{n=1}^NF_n	\\
0&=&\text{Ker}(X)\cdot\underline{\tilde F}	\\
\tilde F&\geq& \underline{F}
\end{eqnarray*}
This is a necessary condition for cloning -- if it cannot be satisfied, cloning is impossible. If there is a satisfying assignment, then the $\beta_n$ need to be derived so that the remaining conditions can be checked. If all are satisfied, cloning is possible, and the fidelities $\tilde F_y$ are attained. If not, the question is unresolved. The remaining conditions do not appear to reduce to linear ones, but are readily specified. If we use $G_{a,b}^{(M,L)}$ as a synonym for $G_y^{(M,L)}$ when $y$ is of weight 2, with $y_a=y_b=1$ ($a\neq b$), and pick any 4 distinct sites ($N\geq 4$), then
\begin{equation}
\underline{\beta}^T(G_{a,b}+G_{c,d}-G_{a,c}-G_{b,d})\underline{\beta}=\frac{2(d-1)(\beta_a-\beta_d)(\beta_b-\beta_c)}{(d+L)\binom{d+L-1}{d}}.	\label{eq:cons2}
\end{equation}
Recalling Eq.\ (\ref{eq:eliminate}), we can easily evaluate terms such as $\beta_a-\beta_d$, and hence we have a whole new set of (non-convex) consistency conditions to verify. If we think of $F_{a,b}+F_{c,d}-F_{a,c}-F_{b,d}$ as an inner product $\underline{v}_{abcd}\cdot\underline{F}$, then one just has to pick $\half N(N-3)$ linearly independent variants of $\underline{v}$, and all necessary conditions have been checked (this is the dimension of the space comprised of all possible vectors $\underline{v}$, and is orthogonal to the vectors $\underline{\tilde v}$ for which $\underline{\tilde v}\cdot\underline{F}=F_n$).

To see that the main (normalization) condition is convex, observe that for sufficiently large $N$, $1-2\sqrt{\alpha\gamma}d-(d-1)(N-2)\gamma$ is certainly negative, and
\begin{lemma} \label{lem:convex}
The function
$$
f_N(x)=\left(\sum_{n=1}^N\sqrt{x_n}\right)^2
$$
is concave.
\end{lemma}
\begin{proof}{
We use a proof by induction to show that $f_N(\alpha x+\beta z)\geq\alpha f_N(x)+\beta f_N(z)$. For the base case, we examine $N=1$:
$$
\left(\sqrt{\alpha x_1+\beta z_1}\right)^2\geq\alpha(\sqrt{x_1})^2+\beta(\sqrt{z_1})^2,
$$
which is straightforward. Now we make the inductive step. Consider
$$
f_N(\alpha x+\beta z)=\left(\sqrt{f_{N-1}(\alpha x+\beta z)}+\sqrt{\alpha x_N+\beta z_N}\right)^2,
$$
assuming that $f_{N-1}(\alpha x+\beta z)\geq\alpha f_{N-1}(x)+\beta f_{N-1}(z)$. We simply need to show that the left-over terms are positive, i.e.
\begin{equation*}
\sqrt{(\alpha x_N+\beta z_N)f_{N-1}(\alpha x+\beta z)}\geq \alpha\sqrt{f_{N-1}(x)x_N}+\beta\sqrt{f_{N-1}(z)z_N}
\end{equation*}
Square, and again apply convexity. This is equal to
$$
\left(\sqrt{x_Nf_{N-1}(z)}-\sqrt{z_Nf_{N-1}(x)}\right)^2\geq 0,
$$
which is clearly true.
}\end{proof}

\subsection{\texorpdfstring{$2\leq M\leq N-2$}{Other M} Cloning}\label{sec:otherM}

For $M=1,N-1$, we can completely solve certain special values of $L$, and are left with only a modest number of constraints to verify in other cases (computationally, resolution of whether cloning is possible may still be a hard problem, but the more constrained it is, the more effectively we can witness the feasibility of cloning). However, for other values of $M$, there is no equivalent to Lemma \ref{lem:linear}.

\begin{lemma} \label{lem:nolinear}
For $(M,L,N)$ cloning then if $N\leq 2M$ or $M$ is even, and $N>M+1$, there are no linear combinations
$$
P:=g_0+\sum_{y\in\Lambda}g_yG_y
$$
which are non-trivial rank 1 projectors.
\end{lemma}
\begin{proof}{
Central to our proof is the observation that for any 4 bit strings $z_1$ to $z_4$ of equal weight ($M$), if $z_1\cup z_2=z_3\cup z_4$, then $\bra{z_1}P\ket{z_2}=\bra{z_3}P\ket{z_4}$ because it must be that $z_1\cdot z_2=z_3\cdot z_4$ and $\bar z_1\cap\bar z_2=\bar z_3\cap\bar z_4$.

Assume that $P=\proj{v}$ for some vector $\ket{v}$. For any state $\ket{\psi}$, $P\ket{\psi}\propto\ket{v}$. Select two different basis states $z_1$, $z_2$ such that $\half M\leq z_1\cdot z_2<M$. We know that
$$
P\ket{z_1}\propto P\ket{z_2}.
$$
However, look at a basis element $\ket{x}$ for which $\bar{x}\cdot(z_1\oplus z_2)=0$ (i.e.\ $x\cup z_1=x\cup z_2$). By our observation, $\bra{x}P\ket{z_1}=\bra{x}P\ket{z_2}$. Hence, either $\bra{x}P\ket{z_1}=0$ (i.e.\ either $P\ket{x}=0$ or $P\ket{z_2}=0$ for $P$ to be rank 1) or $P(\ket{z_1}-\ket{z_2})=0$.

There must be at least one value $z_1$ for which $\braket{v}{z_1}\neq 0$, so start there. It is either that $\braket{v}{z_1}=\braket{v}{z_2}$, or $P\ket{x}=0$ for all compatible $x$. Pick a specific $x$ for which $x\cdot(z_1\cup z_2)=M$ and $\bar x\cdot(z_1\oplus z_2)=0$. There is always at least one such $x$. This choice means that $z_1\cup x=z_1\cup z_2$ so that, by our observation,
$$
\bra{z_1}P(\ket{z_2}-\ket{x})=0,
$$
but now we know that $P\ket{z_1}\neq 0$ by assumption, so either $P\ket{z_2}=P\ket{x}=0$ or $\braket{v}{z_1}=\braket{v}{z_2}=\braket{v}{x}$. Ultimately, this propagates -- either $\braket{v}{z_1}=\braket{v}{z_2}$ for all $z_2: z_2\cdot z_1\geq \half M$ or $P\ket{z_2}=0$ for all $z_2: \half M\leq z_2\cdot z_1<M$.

However, we can now use our observation again. For any arbitrary $z_2$, and the $z_1$ that we fixed, if $M$ is even or $N\leq 2M$, there always exist $z_3$ and $z_4$ such that $z_1\cup z_2=z_3\cup z_4$ and $z_1\cdot z_3\geq\half M$, $z_1\cdot z_4\geq\half M$ (If $M$ is odd and $N>2M$, then for $z_2: z_1\cdot z_2=0$, there are no suitable choices of $z_3,z_4$). The only possible solutions to this are that for some $t: w_x=K$ ($K\geq M$),
$$
\ket{v}=\frac{1}{\sqrt{\binom{K}{M}}}\sum_{\substack{z:w_z=M\\ z\cdot t=M}}\ket{z}.
$$

Of course, we already know that in the case of $N=M+1$, it is possible to find linear combinations for which $K=M$ (i.e.\ a projector on a single basis state). So, we aim to show that this is impossible for larger values of $N$. Consider an arbitrary permutation $\pi$ of the $N$-bit strings. If we take a sum of the relation
$$
\proj{v}=g_0G_{\underline 0}+\sum_yg_yG_y
$$
over all such permutations, we find that
$$
\frac{1}{\binom{N}{K}\binom{K}{M}}\sum_{x,z}\ket{x}\bra{z}{\textstyle\binom{N-2M+x\cdot z}{N-K}}=g_0G_{\underline 0}+\frac{\sum_yg_y}{\binom{N}{M}}G_{\underline 0}^{(M,L)}.
$$
All these matrices have matrix elements which depend only on the values $x\cdot z$, and we thus have a set of simultaneous equations
$$
\frac{\binom{K-M}{M-q}}{\binom{N-M}{M-q}}=\frac{1}{\binom{M+d-1-q}{d-1}}\tilde g_0+\frac{1}{\binom{2M+d-1-1}{d-1}}g_L
$$
to be satisfied, for $q=\max(0,2M-N)\ldots M$. Note that, under the conditions of the theorem ($N>M+1$ and $M\geq 2$), this means that there are at least 3 separate equations to satisfy, and only two free parameters to select, and thus cannot be solved in general. Indeed, since $K$ only appears on the LHS of the equation, it must be that for any given $M,N$, there is no more than one compatible value of $K$.  For simplicity, we have assumed that $L=M$, and have made the replacements $\tilde g_0=\binom{N}{M}g_0$, $g_L=\binom{N+d-1}{M}\sum_yg_y$.

In the case where $N\leq 2M$, we start by considering the possibility that $K\leq 2M-2$ (the choice of $L$ is irrelevant here). In this case, $q=2M-K-1$ and $q=2M-K-2$ are both valid values, and mean that the left-hand side of the equation is 0 in both cases. Hence, $\tilde g_0=g_L=0$. This is clearly incompatible with any instance where the left-hand side is non-zero, such as $q=M$. Only the special cases of $K=N-1,N$ remain. For $K=N$, take the cases $q=M,M-1,M-2$, and solve simultaneously between the equation pairs $M,M-1$ and $M-1,M-2$ to derive two possible solutions for $g_L$. These are
$$
\frac{g_LM}{d}=\binom{M+d}{d}=\binom{M+d-1}{d},
$$
which are clearly never equal, so there is no satisfying assignment. The equivalent expression for $K=N-1$ is
\begin{equation*}
\begin{split}
\frac{g_LM(d-1)(N-M)}{d\binom{M+d-1}{d}}&=\frac{M+d}{M}((N-M)(d-1)-d)	\\
&=(N-M-1)(d-1)-d-1,
\end{split}
\end{equation*}
which would require
$$
d=\frac{N-2M-1}{N-M-2},
$$
but $d\geq 2$, so this cannot happen.

To analyse larger values of $N$, we first observe that our existing arguments automatically cover the case $K\leq 2M-2$. We can thus restrict to the range $2M-1<K$. In a similar vein to before, we consider cases of $q$ in pairs, but now it's $q=0,1$ and $q=1,2$, which we are assured exist due to $N>2M$. We find that
\begin{eqnarray*}
g_L\frac{M(d-1)\binom{K-M}{M-1}}{\binom{N-M}{M-1}\binom{2M+d}{d}d}&=&(M+1)\left(\frac{K-2M+1}{N-2M+1}-\frac{M+d}{M+1}\right)	\\
&=&\frac{(2M+d+1)(M+2)}{2M+1}\left(1-\frac{M+d+1}{M+2}\frac{N-2M+2}{K-2M+2}\right).
\end{eqnarray*}
This equality can be rewritten as
\begin{equation*}
\frac{K-2M+2}{N-2M+1}\frac{2M+1}{2M+d+1}-1=
\frac{N-M+d-1}{(M+d)(N-2M+1)-(M+1)(K-2M+1)}.
\end{equation*}
The left-hand side is negative, while the right is positive. Satisfaction is impossible.
}\end{proof}

However, we do still benefit from a generalization of Lemma \ref{lem:consistent}.

\begin{lemma}\label{lem:consistent2}
Let $2M< L<N-2M$. Define the $\binom{N}{2M}\times\binom{N}{L}$ matrix $X$ as $\bra{x}X\ket{y}=\delta_{x\cdot y=2M}$ where $x,y\in\{0,1\}^N$ and $w_x=2M, w_y=L$. Any $\underline{v}\in\text{Ker}(X)$ satisfies $\sum_{y\in\Lambda}v_yG_y=0$.
\end{lemma}
\begin{proof}{
Again, if $Y=\sum_{y\in\Lambda}v_yG_y$, then we must consider the diagonal, $\bra{x}Y\ket{x}$ and off-diagonal elements $\bra{x}Y\ket{z}$. We have that
\begin{eqnarray*}
\bra{x}Y\ket{x}&=&\sum_{q=0}^M\frac{1}{\binom{M+L+d-q-1}{d-1}}\sum_{y:y\cdot x=q}v_y	\\
\bra{x}Y\ket{z}&=&\sum_{q=0}^{2M-x\cdot z}\frac{1}{\binom{M+L+d-q-1-x\cdot z}{d-1}}\sum_{y:y\cdot (x\cup z)=q}v_y,	\\
\end{eqnarray*}
so all we need to know is that for all $q=0,1,\ldots, 2M$, $\sum_{y:y\cdot x=q}v_y=0$ when $x\in\{0,1\}^N: w_x=q$. Pick any such $x$, then
$$
\bra{x}X\ket{v}=0=\sum_{y:x\cdot y=2M}v_y
$$
by definition. This proves a base case for induction. Now assume that
$$
0=\sum_{y:x\cdot y=q}v_y
$$
for all $k+1\leq q\leq 2M$, and we aim to prove it for the value $k$. Consider
\begin{eqnarray*}
0&=&\sum_{\substack{z\in\{0,1\}^N\\w_z=2M\\z\cdot x=k}}\bra{z}X\ket{v}=\sum_{z:z\cdot x=k}\sum_{y:z\cdot y=2M}v_y\\
&=&\sum_{q=k}^{2M}{\textstyle{\binom{2M}{k}\binom{N-2M}{2M-k}\binom{2M-k}{q-k}\binom{N-4M+k}{L-2M+k-q}}}\sum_{y:y\cdot x=q}v_y.
\end{eqnarray*}
Using our prior assumption, the desired result is clear, and the inductive step is proven.
}\end{proof}

Hence, our general problem is reduced from $\binom{N}{L}$ quadratic constraints to $\binom{N}{2M}$ quadratic ones, with the difference being made up by linear constraints.

\section{Computational Complexity of Cloning}

Ultimately, we would like to be able to answer the question ``Given a set of fidelities $\{F_y\}$ for $y\in\Lambda\subseteq\{0,1\}^N$, is cloning possible with these fidelities?". This yes/no question lends itself to an analysis of its computational complexity. For any family of sets $\Lambda(N)$ for increasing $N$, a solution is easily verified -- we can be given a proof in the form of a set of $\{\beta_x\}$, and all we have to do is verify the normalization condition, and evaluate the achieved fidelities. The run time of such a check is polynomial in the size of the problem instance $|\Lambda|$. (Note that, unless $M$ is finite, this does not necessarily mean that the run time is polynomial in $N$.) Thus, the problem is contained within the complexity class NP.

The fact that our computational problem is necessarily phrased as a non-convex optimization problem suggests that it is a hard problem (assuming P$\neq$NP), and we conjecture that the problem is NP-complete, but have no proof. While the matrices of Eq.\ (\ref{eq:cons2}) are reminiscent of those in \cite{pardalos1991}, in fact it is not even clear that large $N$ scaling is an interesting question, because we know that if all clones have to have non-trivial fidelities, the optimal strategy must tend towards the best classical strategy (measure and reproduce the measurement result) on all but a finite subset of the inputs and outputs (this is reminiscent of the behaviour emerging from the case study of $2\rightarrow 4$ cloning of qubits in section \ref{sec:casestudy}). Of course, there may be subtleties of the attainable fidelities, but these could easily lie at the the limit of being so subtle as to be irrelevant. As such, the issue of computational complexity remains open.

\section{Quantum Circuits for Asymmetric Cloning} \label{sec:circuits}

Let us assume that we have found that a particular set of fidelities is achievable, and hence have a corresponding set of parameters $\{\beta_x\}$. How is cloning implemented? In principle, we know the unitary operation that we could implement -- it's specified in Lemma \ref{lem:upper}, but we don't have an explicit circuit construction for it to show that it can be efficiently implemented.

Previously \cite{wang2011}, it has been suggested that the optimal cloner could be implemented by applying a superposition of different swap operations: start with the state $\ket{B^{(M)}}\ket{\Phi}$ and, with amplitude $\beta_x$, swap the spins 1 to $M$ with those specified by the bit string $x$ in order to produce $\sum_x\beta_x\ket{B_{x}}\ket{\Phi}_{\bar x}$. However, the construction of Wang et al. \cite{wang2011} was only a sketch: it did not contain an efficiency analysis.
Moreover, the implementation was necessarily probabilistic.
This is inappropriate for quantum cloning because, with only  one set of input states, there is no option to repeat until successful. For this reason, we do not give an explicit construction, but it is closely related to the next construction.

In any case, a far better construction is to make the state $\ket{\chi}$ and to teleport the input states into the input spins. While this will use a similar technique to create the state, the advantage is that the probabilistic part of the algorithm can be allowed to fail as it doesn't affect the state to be cloned, and can therefore be run using a repeat-until success strategy. The teleportation protocol runs as follows: one starts with the $M$ input copies of state $\ket{\psi}$ on system $\text{IN}'$, and implements a Bell measurement on the symmetric subspaces of $\text{IN}'$ and $\text{IN}$, i.e.\ projecting onto the basis
$$
\ket{\psi_{ab}}=\frac{1}{\sqrt{\binom{M+d-1}{M}}}\!\!\!\!\sum_{i=0}^{\binom{M+d-1}{d-1}-1}\!\!\!\!\omega^{ia}\ket{\phi_i^\text{IN'}}\ket{\phi^\text{IN}_{i+b\text{ mod }\binom{M+d-1}{M}}}
$$
where $\omega=e^{2\pi i/\binom{M+d-1}{M}}$ and $a,b=0,1,\ldots \binom{M+d-1}{M}-1$. If the result is $a=b=0$, then the teleportation protocol works exactly as intended, and the clones appear on the output space $\text{OUT}$. Otherwise, a correction operation is required which maps states $$\sum_x\beta_x\omega^{-ia}\ket{\phi^x_{i+b\text{ mod }\binom{M+d-1}{M}}}\ket{\Phi}_{\bar x}\mapsto\sum_x\beta_x\ket{\phi^x_{i}}\ket{\Phi}_{\bar x}.$$ This is unitary by construction.

The required unitary is closely related to that for the unitary implementation of cloning and, as such, we do not have a general method for its construction (or an efficiency analysis). However, for the special case of $M=1$, the corrective gates are particularly simple -- each different Bell state projection effectively implements teleportation with a different single-qubit rotation applied to it (for $d=2$, these are just the standard Pauli operators), and so the cloned states are just the ideal clones with those same rotations applied. They can therefore be removed by transversal application (i.e.\ simultaneous application to each of the $N$ output spins) of the corresponding inverse operation. Indeed, for general $M$ we can compensate for the different $a$ answers in this way, although the $b$ answers are not so easily compensated. However, even if we just proceeded to post-select on a $b=0$ result, there would be a success probability of $1/\binom{M+d-1}{d-1}$, which is better than the success rate of the previous method.

To construct the desired state $\ket{\chi}$ (in fact, we will give the construction for the mixed state described in Corollary \ref{cor:1}), we first produce a state $\ket{\chi_0}$:
$$
\frac{\sum_{i=0}^{\binom{M+d-1}{d-1}-1}\sum_{j=0}^{\binom{N-M+d-1}{d-1}-1}\!\!\!\!\!\ket{\phi_i^\text{IN}}\ket{\phi_i^M}\ket{\phi_j^{N-M}}\ket{\phi_j^{N-M}}}{\sqrt{\binom{M+d-1}{d-1}\binom{N-M+d-1}{d-1}}}
$$
where the last set of $N-M$ spins are ancilla systems $A$ that one should trace over to return the desired mixed state (although keeping the ancillas is probably useful for applying the compensatory rotations to account for the different measurement results). This is easy since we can start from a GHZ-like state and can convert between computational basis states and symmetric states \cite{bacon2007}. Next, we need to produce a state
$$
\ket{\beta}=\frac{\sum_x\beta_x\ket{x_1}\ket{x_2}\ket{x_3}\ldots\ket{x_M}}{\sqrt{\sum_x\beta_x^2}},
$$
where we use $\ket{x_n}$ to denote a set of $\lceil\log_2N\rceil$ qubits containing a binary representation of the value $m$ where the $m$th bit of $x$ is the $n$th 1 in the string (to avoid confusion when swapping different sites, one should make the small modification that any bits $m$ in $x$ which are 1 and $m\leq M$ must be specified as $x_m=m$). For fixed $M$, since there are only $\binom{N}{M}$ terms $\beta_x$, the probability distribution $\beta_x^2$ can be efficiently integrated. Thus, $\ket{\beta}$ is easily constructed \cite{grover2002}. Next, we take $\ket{\beta}\ket{\chi_0}$ and apply controlled-swaps controlled off the spins of the $\ket{\beta}$ system. In essence, if $\ket{x_n}$ takes value $m$ then apply a swap between spins $n$ and $m$ of the output space in system $\ket{\chi_0}$. Therefore, we have produced the state
\begin{equation*}
\frac{\sum_x\beta_x\ket{x_1}\ket{x_2}\ket{x_3}\ldots\ket{x_M}}{\sqrt{\sum_x\beta_x^2}}\times \frac{\sum_{i=0}^{\binom{M+d-1}{d-1}-1}\sum_{j=0}^{\binom{N-M+d-1}{d-1}-1}\ket{\phi_i^{\text{IN}}}\ket{\phi_i^x}\ket{\phi_j^{\bar x}}\ket{\phi_j^{A}}}{\sqrt{\binom{M+d-1}{d-1}\binom{N-M+d-1}{d-1}}}.
\end{equation*}
Now we project all the $\ket{\beta}$ qubits onto the state $\ket{+}=(\ket{0}+\ket{1})/\sqrt{2}$. If successful, the output state is that desired,
$$
\sum_x\beta_x\frac{\sum_{i=0}^{\binom{M+d-1}{d-1}-1}\sum_{j=0}^{\binom{N-M+d-1}{d-1}-1}\ket{\phi_i^\text{IN}}\ket{\phi_i^x}\ket{\phi_j^{\bar x}}\ket{\phi_j^A}}{\sqrt{\binom{M+d-1}{d-1}\binom{N-M+d-1}{d-1}}},
$$
and this happens with a probability
$$
\frac{1}{N^M\sum_x\beta_x^2}.
$$
Since $1=\sum_{x,z}\beta_x\beta_z/\binom{M+d-1-x\cdot z}{d-1}\geq\sum_x\beta_x^2$, this probability is no smaller than $N^{-M}$. For fixed $M$, a polynomial number of repetitions produces the target state.

\section{Conclusions} \label{sec:conclusion}

This paper has demonstrated necessary and sufficient conditions for optimal $1\rightarrow N$ ($L=1,N-1$) and $N-1\rightarrow N$  universal cloning. In principle, these conditions can be used for bounding many-body correlations in a quantum system. We also explored why the cases of $2\leq M\leq N-2$ are more challenging -- the conditions are necessarily formulated as non-convex constraints (whereas those that we have been able to solve can be reduced to linear constraints, which are consequently convex). We conjecture that, in such cases, the cloning problem is NP-complete to resolve, and anticipate that the formulation provided in this paper should prove a suitable starting point for such studies of the computational complexity.

\appendix

\label{app:1}

We study the spectral properties of the matrices $G_{\underline 0}^{(M)}$ which, while only tangentially relevant to the main text, may prove useful for future investigations.

\begin{lemma}\label{lem:evalues1}
For a fixed $N$, consider a matrix
$$
G^{(M)}=\sum_{x,z:w_x=w_z=M}f_{x\cdot z}^{(M)}\ket{x}\bra{z}.
$$
If $\lambda$ is an eigenvalue of $G^{(M)}$ of degeneracy $g$ and
$$
\tilde\lambda=\frac{(M+1-k)f_{k+1}^{(M+1)}+(N-2M-1+k)f_k^{(M+1)}}{(M+1-k)f_k^{(M)}+kf_{k-1}^{(M)}}
$$
is independent of $k$, then $\lambda\tilde\lambda$ is an eigenvalue of $G^{(M+1)}$ with degeneracy $g$. If $N>2M$, then there are $\binom{N}{M+1}-\binom{N}{M}$ additional (degenerate) eigenvalues given by
$$
\frac{\binom{N}{M+1}f_{M+1}^{(M+1)}-\tilde\lambda\binom{N}{M}f_M^{(M)}}{\binom{N}{M+1}-\binom{N}{M}}.
$$
Recursive calculation of the eigenvalues starts from $M=0$ with a single eigenvalue of $f_0^{(0)}$.
\end{lemma}
\begin{proof}{
Let $\ket{\lambda^{(M)}}=\sum_{y:w_y=M}\lambda_y\ket{y}$ satisfy $G^{(M)}\ket{\lambda^{(M)}}=\lambda\ket{\lambda^{(M)}}$. We aim to prove that
$$
\ket{\lambda^{(M+1)}}=\sum_{\substack{x\in\{0,1\}^N\\w_x=M+1}}\sum_{\substack{y\in\{0,1\}^N\\w_y=M\\x\cdot y=M}}\lambda_y\ket{x}
$$
satisfies $G^{(M+1)}\ket{\lambda^{(M+1)}}=\tilde\lambda\lambda\ket{\lambda^{(M+1)}}$. Since $\ket{\lambda^{(M)}}$ is an eigenvector of $G^{(M)}$ then,
$$
\lambda \lambda_y=\sum_{\tilde x:w_{\tilde x}=M}\lambda_{\tilde x}f^{(M)}_{\tilde x\cdot y}\qquad \forall y:w_y=M.
$$
Selecting an $x$ with $w_x=M+1$, then
$$
\lambda\sum_{\substack{y:w_y=M\\x\cdot y=M}}\lambda_y=\sum_{\substack{y:w_y=M\\x\cdot y=M}}\sum_{\tilde x:w_{\tilde x}=M}\lambda_{\tilde x}f^{(M)}_{\tilde x\cdot y}
$$
Upon performing the sum over $y$ first, we have to consider that the difference between the string $x$ and any choice of $y$ is a single site (which is a 1 for string $x$ and a 0 for $y$), so there are $M+1$ different strings $y$, but this means that $x\cdot\tilde x$ and $y\cdot\tilde x$ must either be the same, or differ by 1. Hence,
\begin{equation}
\lambda\sum_{\substack{y:w_y=M\\x\cdot y=M}}\lambda_y=\!\!\!\sum_{\tilde x:w_{\tilde x}=M}\!\!\!\!\lambda_{\tilde x}\left((M+1-x\cdot\tilde x)f^{(M)}_{\tilde x\cdot x}+x\cdot\tilde x f^{(M)}_{\tilde x\cdot x-1}\right)	\label{eqn:M}
\end{equation}
With this relation in place, we can proceed to look at the case of $M+1$. We need to prove that for all $x: w_x=M+1$,
$$
\lambda\tilde\lambda\!\!\!\!\sum_{\substack{y: w_y=M\\ x\cdot y=M}}\lambda_y\!=\!\!\!\!\sum_{z:w_z=M+1}\sum_{\substack{y:w_y=M\\z\cdot y=M}}f^{(M+1)}_{x\cdot z}\lambda_y
$$
We reorder the two sums,
\begin{eqnarray*}
\text{RHS}&=&\sum_{y:w_y=M}\sum_{\substack{z:w_z=M+1\\z\cdot y=M}}\lambda_yf^{(M+1)}_{x\cdot z}	\\
&=&\!\!\!\!\sum_{y:w_y=M}\!\!\!\!\lambda_y\left((M+1-x\cdot y)f^{(M+1)}_{x\cdot y+1}+(N-2M-1+x\cdot y)f^{(M+1)}_{x\cdot y}\right)
\end{eqnarray*}
Provided $\tilde\lambda$ is independent of the value $k\equiv x\cdot y$, this is as desired.

When $N>2M$, increasing the value of $M$ increases the number of eigenvalues. If all the additional eigenvalues take on the same value, this value must be given by
$$
\frac{\Tr(G^{(M+1)})-\tilde\lambda\Tr(G^{(M)})}{\binom{N}{M+1}-\binom{N}{M}}=\frac{\binom{N}{M+1}f_{M+1}^{(M+1)}-\tilde\lambda\binom{N}{M}f_{M}^{(M)}}{\binom{N}{M+1}-\binom{N}{M}}
$$
Proving that all the new eigenvalues are the same requires more careful consideration. Define the matrix
$$
G_T=\sum_{x,z\in\{0,1\}^N}\delta_{w_x,w_z}f^{(w_x)}_{x\cdot z}\ket{x}\bra{z}\equiv\bigoplus_{M=0}^NG^{(M)},
$$
which effectively corresponds to a Hilbert space of $N$ qubits. Obviously, the different values of $w_x$ (the separate $G^{(M)}$) define excitation subspaces, i.e.\ $G_T$ commutes with the $J_Z$ operator for $N$ qubits. Furthermore, $G_T$ is invariant under permutations of those $N$ qubits: for a permutation $\pi$ acting on bit strings, $(\pi x)\cdot(\pi z)=x\cdot z$. So, $G_T$ also commutes with the total angular momentum operator $J^2$, and we know it must therefore decompose into a structure of the form (given explicitly for even $N$ \cite{bartlett2007})
$$
\bigoplus_{j=0}^{N/2}\mathcal{M}_j\otimes\identity
$$
where the $\identity$ term associated with the index $j$ is of dimension
$$
\binom{N}{\frac{N}{2}-j}-\binom{N}{\frac{N}{2}-j-1}.
$$
So, we can clearly identify the `new' eigenvalues appearing for a particular value of $M=\half N-j$ as being the first instance of the $\mathcal{M}_j$ subsystem (which populates the $w_x=\half N-j$ to $\half N+j$ excitation subspaces), and therefore have the correct degeneracy to all have the same eigenvalue.
}\end{proof}

\begin{corollary} \label{cor:evalues}
The eigenvalues of matrix $G_{\underline 0}^{(M)}$ are 
$$
\frac{\binom{d-2+k}{k}\binom{N+d-1}{M}}{\binom{M+d-1}{M}\binom{N+d-1}{k}}\qquad k=0,\ldots M
$$
with degeneracy $\binom{N}{k}-\binom{N}{k-1}$.
\end{corollary}
\begin{proof}{
With $f^{(M)}_k=\frac{1}{\binom{M+d-1-k}{d-1}}$, it turns out that $\tilde\lambda=\frac{N-M-1+d}{M+d}$, and the conditions of Lemma \ref{lem:evalues1} are satisfied.
}\end{proof}

\begin{corollary} \label{cor:junk}
The inverse of $G_{\underline 0}^{(M)}$ is given by
$$
G_{\underline 0}^{-1}=\frac{(d+M-1)(N+d-M-1)}{(d-1)(d+N-1)}\sum_{x,z}\frac{(-1)^{M+x\cdot z}\ket{x}\bra{z}}{\binom{d+N-2}{M-x\cdot z}}.
$$
\end{corollary}
\begin{proof}{
Again, the conditions of Lemma \ref{lem:evalues1} hold, now with $\tilde\lambda=\frac{M+d-1}{N+d-2-M}$, the inverse of the scale factor for $G_{\underline 0}^{(M)}$.
}\end{proof}


\bibliography{../../../References}

\end{document}